\documentclass[a4paper,UKenglish,cleveref, autoref, thm-restate]{lipics-v2021}

\title{Semantics for a Turing-complete Reversible
Programming Language with Inductive Types} 

\author{Kostia Chardonnet}{Department of Computer Science and Engineering, University of Bologna, Italy}{}{}{is partially supported by the MIUR FARE project CAFFEINE, ``Compositional and Effectful Program Distances'', R20LW7EJ7L.}

\author{Louis Lemonnier\footnote{Corresponding author.}}{Université Paris-Saclay, CNRS, ENS Paris-Saclay, Inria, Laboratoire Méthodes Formelles, 91190, Gif-sur-Yvette, France}{}{https://orcid.org/0000-0003-1761-3244}{}

\author{Benoît Valiron}{Université Paris-Saclay, CNRS, CentraleSupélec, ENS Paris-Saclay, Inria, Laboratoire Méthodes Formelles, 91190, Gif-sur-Yvette, France}{}{}{}

\authorrunning{K. Chardonnet, L. Lemonnier and B. Valiron} 

\titlerunning{Sem. for a Turing-complete Rev. Prog. Lang. with Induct. Types}

\Copyright{Kostia Chardonnet, Louis Lemonnier and Benoît Valiron} 



\category{} 

\relatedversion{} 

\funding{This work has been partially funded by the French National
  Research Agency (ANR) by the projects ANR-21-CE48-0019,
  ANR-19-CE48-0014, ANR-22-CE47-0012, and within the framework of
  ``Plan France 2030'', under the project ANR-22-PETQ-0007,
  ANR-22-PNCQ-0001.}

\acknowledgements{
	The authors thank Vladimir Zamdzhiev for his expert insight on
	specific parts of the denotational semantics.
}

\nolinenumbers 

\EventEditors{John Open, Jena Access}
\EventNoEds{1}
\EventLongTitle{23th Conference on Very Interesting Topics 2024}
\EventShortTitle{CVIT 2024}
\EventAcronym{CVIT}
\EventYear{2024}
\EventDate{2024}
\EventLocation{Online}
\EventLogo{}
\SeriesVolume{$\infty$}
\ArticleNo{$\infty$}

\usepackage{amsmath}
\makeatletter
\newcommand{\allignLabel}[1]{\refstepcounter{equation}(\theequation)\def\tmplab{#1}\ltx@label\tmplab}
\makeatother
\usepackage{anyfontsize}
\usepackage{amsthm}
\usepackage{amssymb}
\usepackage{proof}
\usepackage{hyperref}
\usepackage{amssymb}
\usepackage{mathrsfs}
\usepackage{stmaryrd}
\usepackage{dsfont}
\usepackage{cleveref}
\usepackage[all]{hypcap}

\newtheorem{property}[theorem]{Property}

\usepackage{tikz}

\usepackage{todonotes}
\usepackage{regexpatch}

\newcommand{\N}{\mathbb{N}}

\newcommand{\CC}{\mathcal C}
\newcommand{\DD}{\mathcal D}

\newcommand{\ov}[1]{\ensuremath{\overline{#1}}}
\newcommand{\opn}[1]{\ensuremath{\operatorname{#1}}}

\newcommand{\set}[1]{\ensuremath{\{#1\}}}
\newcommand{\mynl}{\\[1ex]}
\newcommand{\natS}{\ensuremath{\mathbb{N}}}
\newcommand{\natT}{\ensuremath{\mathtt{nat}}}

\newcommand{\boolT}{\ensuremath{\mathbb{B}}}

\newcommand{\one}{\ensuremath{\mathds{1}}}
\newcommand\1\one

\newcommand{\intf}[1]{\left\llbracket #1 \right\rrbracket}
\newcommand{\interp}[1]{\ensuremath{\intf{#1}}}

\newcommand{\unit}{\ensuremath{()}}
\newcommand{\alt}{~\mid~}
\newcommand{\inl}[1]{\ensuremath{\mathtt{inj}_\ell{\;#1}}}
\newcommand{\inr}[1]{\ensuremath{\mathtt{inj}_r{\;#1}}}
\newcommand{\ini}[1]{\ensuremath{\mathtt{inj}_i{\;#1}}}
\newcommand{\pv}[2]{\ensuremath{\langle #1,#2 \rangle}}
\newcommand{\pair}[2]{\ensuremath{\langle #1,#2 \rangle}}

\newcommand{\fold}[1]{\ensuremath{\mathtt{fold}{\;#1}}}

\newcommand{\lett}{\ensuremath{\mathtt{let}}}
\newcommand{\letv}[3]{{\mathtt{let}}\,{#1}={#2}~{\mathtt{in}}~{#3}}

\newcommand{\tc}{\mathtt{t}\!\mathtt{t}}
\newcommand{\fc}{\mathtt{f}\!\mathtt{f}}

\newcommand{\entailiso}{\vdash_{\isoterm}}

\newcommand{\entail}{\vdash}

\newcommand{\iso}{\ensuremath{\leftrightarrow}}
\newcommand{\isovar}{\ensuremath{\phi}}
\newcommand{\isolambdavar}{\ensuremath{\psi}}

\newcommand{\isobasique}{\ensuremath{\set{v_1\iso e_1\mid \dots \mid v_n \iso e_n}}}
\newcommand{\isobreduit}{\ensuremath{\set{v_i\iso e_i}_{i \in I}}}

\newcommand{\fix}{\ensuremath{\mathtt{fix}~}}
\newcommand{\ffix}{\ensuremath{\mathtt{fix}~}}
\newcommand{\nfix}[1]{\ensuremath{\mathtt{fix}^{#1}~}}
\newcommand{\isoterm}{\ensuremath{\omega}}

\newcommand{\ttt}{\ensuremath{\mathbf{tt}}}
\newcommand{\fff}{\ensuremath{\mathbf{ff}}}
\newcommand{\dup}{\ensuremath{\opn{Dup}}}
\newcommand{\garRem}[2]{\ensuremath{\mathtt{GarbRem}(#1, #2)}}
\newcommand{\gctx}{\ensuremath{\mathcal{G}}}
\newcommand{\match}[3]{#1(#2) = #3}
\newcommand{\finto}{\underset{\mathrm{fin}}{\to}}
\newcommand{\isoto}{\underset{\mathrm{iso}}{\to}}
\newcommand{\termto}{\underset{\mathrm{term}}{\to}}
\newcommand{\subfin}{\triangleleft}
\usepackage{mathtools}
\DeclarePairedDelimiter\floor{\lfloor}{\rfloor}

\newcommand{\len}{\ensuremath{\mathtt{len}}}
\newcommand{\erase}{\ensuremath{\mathtt{erase}}}
\newcommand{\growth}{\ensuremath{\mathtt{growth}}}
\newcommand{\snoc}{\ensuremath{\mathtt{snoc}}}
\newcommand{\snocc}{\ensuremath{\mathtt{snoc'}}}
\newcommand{\It}{\ensuremath{\mathtt{It}}}
\newcommand{\rmBlank}{\ensuremath{\mathtt{rmBlank}}}
\newcommand{\rev}{\ensuremath{\mathtt{rev}}}
\newcommand{\cleanUp}{\ensuremath{\mathtt{cleanUp}}}

\newcommand{\entaile}{\vdash}

\newcommand{\res}[1]{\overline{#1}}

\newcommand{\DCPO}{\mathbf{DCPO}}
\newcommand{\dcpo}{\mathbf{DCPO}}

\newcommand{\PInj}{\mathbf{PInj}}
\newcommand{\Hom}{\text{Hom}}

\newcommand{\inv}{^{\circ}}
\newcommand{\pinv}{^{\circ}}
\newcommand{\iid}{\mathrm{id}}
\newcommand{\comp}{\mathrm{comp}}
\newcommand{\rmcurry}{\mathrm{curry}}
\newcommand{\rmeval}{\mathrm{eval}}
\newcommand{\defeq}{\overset{\text{def.}}{=}}

\newcommand{\den}[1]{\left\llbracket #1 \right\rrbracket}
\newcommand\sem\den
\newcommand{\abs}[1]{\left\vert #1 \right\vert}

\newcommand{\rc}[1]{\res{#1^{\circ}}}

\makeatletter
\newsavebox{\@brx}
\newcommand{\llangle}[1][]{\savebox{\@brx}{\(\m@th{#1\langle}\)}%
\mathopen{\copy\@brx\kern-0.5\wd\@brx\usebox{\@brx}}}
\newcommand{\rrangle}[1][]{\savebox{\@brx}{\(\m@th{#1\rangle}\)}%
\mathclose{\copy\@brx\kern-0.5\wd\@brx\usebox{\@brx}}}
\makeatother

\newcommand{\noma}{\text{\usefont{U}{min}{m}{n}\symbol{'005}}\!}

\DeclareFontFamily{U}{min}{}
\DeclareFontShape{U}{min}{m}{n}{<-> udmj30}{}

\usepackage{tikz}
\usetikzlibrary{quantikz}
\usetikzlibrary{calc}
\usetikzlibrary{arrows.meta}
\usetikzlibrary{shapes, patterns, decorations.pathreplacing, decorations.text}
\usetikzlibrary{trees}
\usetikzlibrary{positioning}
\usetikzlibrary{tikzmark, cd}

\tikzset{
  node rotated/.style = {rotate=180},
  border rotated/.style = {shape border rotate=180},
  downtriangle/.style = {fill=white, draw=black, regular polygon, regular polygon sides=3, border rotated},
  triangle/.style = {fill=white, draw=black, regular polygon, regular polygon sides=3}
}

\usetikzlibrary{decorations.markings,arrows,decorations.pathmorphing}
\usetikzlibrary{shapes.misc,matrix,backgrounds,folding,positioning}
\usetikzlibrary{shapes.geometric}
\usetikzlibrary{shapes.multipart}
\pgfdeclarelayer{edgelayer}
\pgfdeclarelayer{nodelayer}
\pgfsetlayers{background,edgelayer,nodelayer,main}
\tikzset{every path/.style={draw=black!80, line width=0.6pt}}
\tikzstyle{every picture}=[baseline=-0.25em]
\tikzstyle{none}=[inner sep=0mm]
\tikzstyle{box}=[fill=white, draw=black, shape=rectangle]
\tikzstyle{zxnode}=[shape=circle, minimum width=.25cm, inner sep=0.5pt, font=\footnotesize, draw=black,thick]
\tikzstyle{gn}=[zxnode ,fill=green, draw=green!10!black]
\tikzstyle{rn}=[zxnode ,fill=red, draw=red!10!black]
\tikzstyle{H box}=[rectangle,fill=yellow, draw=yellow!10!black,thick,xscale=1,yscale=1,font=\footnotesize,inner sep=1.2pt,minimum width=0.15cm,minimum height=0.15cm]
\tikzstyle{ug}=[regular polygon, regular polygon sides=3, fill=red,draw=black,inner sep = 0pt,minimum width=0.8em]
\tikzstyle{black dot}=[inner sep=0.4mm,minimum width=0pt,minimum height=0pt,fill=black,draw=black,shape=circle]
\tikzstyle{dot}=[black dot]
\tikzstyle{white dot}=[dot,fill=white, inner sep=0.4mm,minimum width=0pt,minimum height=0pt, font=\footnotesize]
\tikzstyle{arrow}=[decoration={markings,mark=at position 1 with
    {\arrow[scale=1.2,>=stealth]{>}}},postaction={decorate}]

\tikzstyle{glabel}=[rounded corners=0.2em,fill=green!30,inner sep=0.1em,font=\scriptsize, anchor=west, xshift=-0.3em, yshift=0,opacity=1]
\tikzstyle{rlabel}=[rounded corners=0.2em,fill=red!30,inner sep=0.1em,font=\scriptsize, anchor=west, xshift=-0.3em, yshift=0,opacity=1]
\tikzstyle{box}=[rectangle, draw=black, fill=white, inner sep=1pt, font=\scriptsize]
\tikzstyle{box-no-outline}=[rectangle, draw=white, fill=white, inner sep=2pt]
\tikzstyle{circle-no-outline}=[circle, draw=white, fill=white, inner sep=0pt]
\tikzstyle{squigglearrow}=[->, line join=round, decorate, decoration={zigzag, segment length=4, amplitude=0.8, post=lineto, post length=2pt}]
\tikzstyle{divide}=[regular polygon, regular polygon sides=3, draw=black, fill=gray!50, inner sep=1.6pt, rounded corners=0.8mm]
\tikzstyle{very thick}=[-, line width=1pt]
\tikzstyle{boxedge}=[draw=gray!50]
\tikzstyle{pbs}=[diamond, draw=black, inner sep=-0.5pt, fill=white]
\tikzstyle{ribbon}=[thick, rounded corners=0.4pt,fill={rgb,255: red,157; green,246; blue,255}, fill opacity=0.7]
\tikzset{tensor/.style={inner sep=2.5pt, draw, circle, path picture={ 
  \draw[black]
(path picture bounding box.south east) -- (path picture bounding box.north west) (path picture bounding box.south west) -- (path picture bounding box.north east);
}}}
\tikzset{plus/.style={inner sep=2.5pt, draw, circle, path picture={ 
  \draw[black]
(path picture bounding box.east) -- (path picture bounding box.west) (path picture bounding box.south) -- (path picture bounding box.north);
}}}
\tikzstyle{contraction}=[circle,draw,font={\scriptsize c}, inner sep= 1pt]
\tikzstyle{parrll}=[circle,draw,font={\scriptsize $\parr$}, inner sep= 0.2pt]
\tikzstyle{small-contraction}=[circle,draw,font={\tiny c}, inner sep= 0.2pt]
\tikzstyle{vacuum}=[rounded rectangle, draw, fill=gray!50, rounded rectangle west arc=none, rotate=180]

\pgfdeclarelayer{edgelayer}
\pgfdeclarelayer{nodelayer}
\pgfsetlayers{background,edgelayer,nodelayer,main}
\tikzstyle{every loop}=[]

\usetikzlibrary{circuits.ee.IEC}
\usepackage{circuitikz}

\newcommand{\tikzfigpathValue}{./figures}

\newcommand{\tikzfig}[1]{%
\IfFileExists{#1.tikz}
  {\input{#1.tikz}}
  {%
    \IfFileExists{\tikzfigpathValue/#1.tikz}
      {\input{\tikzfigpathValue/#1.tikz}}
      {\tikz[baseline=-0.5em]{\node[draw=red,font=\color{red},fill=red!10!white] {\textit{\tikzfigpathValue/#1}};}}%
  }%
}

\usepackage[skins]{tcolorbox}

\newtcolorbox{myframe}[2][]{%
  enhanced,colback=white,colframe=black,coltitle=black,
  sharp corners,boxrule=0.4pt,
  fonttitle=\itshape,
  attach boxed title to top left={yshift=-0.3\baselineskip-0.4pt,xshift=2mm},
  boxed title style={tile,size=minimal,left=0.5mm,right=0.5mm,
    colback=white,before upper=\strut},
  title=#2,#1
}

\begin{document}

\maketitle
\begin{abstract}
	This paper is concerned with the expressivity and denotational semantics of a
	functional higher-order reversible programming language based on Theseus. In
	this language, pattern-matching is used to ensure the reversibility of
	functions. We show how one can encode any Reversible Turing Machine in said
	language. We then build a sound and adequate categorical semantics based on
	join inverse categories, with additional structures to capture
	pattern-matching and to interpret inductive types and recursion. We then
	derive a notion of completeness in the sense that any computable, partial,
	first-order injective function is the image of a term in the language.
\end{abstract}

\section{Introduction}

Originally, reversible computation has emerged as an energy-preserving
model of computation in which no data is ever erased. This comes from
Laundauer's principle which states that the erasure of information is
linked to the dissipation of energy as heat~\cite{Landauer61,
  berut2012experimental}. In reversible computation, given some
process $f$, there always exists an inverse process $f^{-1}$ such that
their composition is equal to the identity: it is always possible to
``\emph{go back in time}'' and recover the input of your
computation. Although this can be seen as very restrictive, non-reversible computation
can be emulated in a reversible setting by keeping track of
intermediate results. As discussed in~\cite{bennett1973logical}, the
simulation of standard computation with reversible computation can be
understood as a notion of \emph{Turing completeness}---provided we
accept that the final result comes together with auxiliary,
intermediate computation.

Reversible computation has since been shown to be a versatile
model. In the realm of quantum computation, reversible computing is at
the root of the construction of \emph{oracles}, subroutines describing
problem instances in quantum algorithms~\cite{nielsen02quantum}. Most
of the research in reversible circuit design can then been repurposed
to design efficient quantum circuits. On the theoretical side,
reversible computing serves the main ingredient in several operational
models of linear logic, whether through token-based Geometry of
Interaction~\cite{mackie1995geometry}
or through the Curry-Howard correspondence for
$\mu$MALL~\cite{chardonnet2022curry, phd-kostia}.

Reversible programming has been approached in two different ways. The
first one, based on Janus and later R-CORE and
R-WHILE~\cite{lutz1986janus,yokoyama2007reversible,
  gluck2019reversible,yokoyama2016fundamentals}, considers imperative
and flow-chart languages. The other one follows a functional
approach~\cite{yokoyama2011reversible,thomsen2015interpretation,james2014theseus,JacobsenKT18,sabry2018symmetric,chardonnet2022curry}: a function $A\to B$ in the language represents a
function ---a bijection--- between values of type $A$ and values of
type $B$. In this approach, types are typically structured, and
functional reversible languages usually feature pattern-matching to
discriminate on values.

One of the issues reversible programming has to deal with is
\emph{non-termination}: in general, a reversible program computes a
\emph{partial injective map}. This intuition can be formalised with
the concept of \emph{inverse
  categories}~\cite{kastl1979inverse,cockett2002restriction-I,cockett2003restriction-II,cockett2007restriction-III}:
categories in which every morphism comes with a partial inverse, for which the
category $\PInj$ of sets and partial injective maps is the emblematic
concrete instance.
This categorical setting has been successfully used in the study of
reversible programming semantics, whether based on
flow-charts~\cite{gluck2017categorical,kaarsgaard2019condition},
with
recursion~\cite{axelsen2016join,kaarsgaard2017join,kaarsgaard2019inversion,kaarsgaard2019engarde},
with side effects~\cite{heunen2015reversible,heunen2018reversible},
\textit{etc}.

Although much work has been dedicated to the categorical analysis of
reversible computation, the \emph{adequacy} of the developed
categorical constructs with reversible functional programming
languages has only recently been under scrutiny, either in
\emph{concrete} categories of partial
isomorphisms~\cite{kaarsgaard2019engarde,kaarsgaard2021join}, or for
simple, \emph{non Turing-complete} languages~\cite{nous21invcat}. A
formal, categorical analysis of a (reversible) Turing-complete, reversible language
is still missing.

\medskip
\noindent
\textsf{\bfseries Contributions.~} In this paper, we aim at closing this gap: we
propose a Turing-complete (understood as in the reversible setting),
reversible language, together with a
categorical semantics. In particular, the contributions of this paper
are as follows.
\begin{itemize} 
	\item A (reversible) Turing-complete, higher-order reversible language with
			inductive types.  Building on the Theseus-based family of languages studied
			in \cite{sabry2018symmetric,nous21invcat,chardonnet2022curry,
			phd-kostia, phd-louis}, we consider an extension with \emph{inductive
			types}, general \emph{recursion}
			and \emph{higher-order} functions.
		\item Sound and adequate categorical semantics. We show how the
			language can be interpreted in join inverse rig categories. The
			result relies on the $\DCPO$-enrichments of join inverse rig
			categories.
		\item A notion of completeness.  We finally discuss how the
			interpretation of the language in the category $\PInj$ is complete
			in the sense that any \emph{first-order} computable, partial
			injective function on the images of types is realisable within
			the language.
\end{itemize}

\section{Language}
\label{sec:language}

In this section, we present a reversible language, unifying and
extending the Theseus-based variants presented in the
literature~\cite{sabry2018symmetric,nous21invcat,
  chardonnet2022curry}. In particular, the language we propose
features higher-order (unlike~\cite{nous21invcat}), pairing,
injection, inductive types (unlike~\cite{sabry2018symmetric}) and
general recursion (unlike~\cite{chardonnet2022curry}). Functions in
the language are based on pattern-matching, following a strict
syntactic discipline: term variables in patterns should be used
linearly, and clauses should be non-overlapping on the left \emph{and}
on the right (therefore enforcing non-ambiguity and
injectivity). In~\cite{sabry2018symmetric,nous21invcat,
  chardonnet2022curry} one also requires exhaustivity for totality. In
this paper, we drop this condition in order to allow non-terminating
behaviour.

\begin{table}[t]
\begin{alignat*}{100}
		&\text{(Base types)} \quad& A, B &&&::=~ && \one \alt A \oplus B
		\alt A \otimes B \alt \mu X. A \alt X \\
    &\text{(Isos)} & T &&&::=&&A\iso B \mid
    T_1 \to T_2 \\[1.5ex]
    &\text{(Values)} & v &&&::=&& \unit \alt x \alt \inl{v} \alt \inr{v}
    \alt \pv{v_1}{v_2} \alt \fold{v} \\
    &\text{(Patterns)} & p &&&::=&& x \alt \pv{p_1}{p_2} \\
    &\text{(Expressions)} & e &&&::=&& v \alt
    \letv{p_1}{\omega~p_2}{e} \\
    &\text{(Isos)} & \isoterm &&&::=&& \isobasique \alt \fix
    \isovar.\isoterm \alt \lambda\isolambdavar.
    \omega \alt \isovar \alt \omega_1~\omega_2\\
    &\text{(Terms)} & t &&&::=&& \unit \alt x \alt \inl{t} \alt \inr{t}
    \alt
    \pv{t_1}{t_2} \alt  \\
    & & &&& && \fold{t} \alt \isoterm~t \alt
    \letv{p}{t_1}{t_2}
\end{alignat*}
\caption{Grammar for terms and types.}
\label{tab:termtypes}
\end{table}

The language is presented in Table~\ref{tab:termtypes}. It consist of
two layers.
\begin{itemize}
\item Base types: The base types consist of the unit type $\one$ along
  with its sole constructor $\unit$, coproduct $A\oplus B$ and
  tensor product $A\otimes B$ with their respective constructors,
  $\inl{(t)}, \inr{(t)}$ and $\pv{t_1}{t_2}$.  Finally, the
  language features inductive types of the form $\mu X. A$ where
  $X$ is a type variable occurring in $A$ and $\mu$ is its
  binder. Its associated constructor is $\fold{(t)}$. The
  inductive type $\mu X. A$ can then be unfolded into
  $A[\mu X. A/X]$, i.e., substituting each occurrence
  of $X$ by $\mu X. A$ in $A$. Typical examples of inductive types
  that can be encoded this way are the natural number, as $\natT
  = \mu X. (\one\oplus X)$ or the lists of types $A$, noted $[A]
  = \mu X. \one\oplus (A\otimes X)$. Note that we only work with
  closed types. We shall denote term-variables with $x, y, z$.

\item Isos types: The language features isos, denoted $\omega$, higher order
	reversible functions whose types $T$ consist either of a pair of base type,
		noted $A\iso B$ or function types between isos, $T_1 \to T_2$. Note that
		the word \emph{iso} comes from isomorphism. However, in this paper, we have
		freed some constraints; in our case, isos are \emph{forward deterministic}
		and \emph{backward deterministic}, meaning that each value has at most one
		image and at most one value that has the former as image. A first-order iso
		of type $A\iso B$ consists of a finite set of \emph{clauses}, written
		$v\iso e$ where $v$ is a value of type $A$ and $e$ an expression of type
		$B$. An expression consists of a succession of applications of isos to some
		argument, described by {\lett} constructions:
	$\letv{(x_1, \dots, x_n)}{\omega~(y_1, \dots, y_n)}{e}$.
	Isos can take other isos as arguments through the
	$\lambda \phi . \omega$ construction.
	Finally, isos can also represent \emph{recursive computation} through the $\ffix
	\isovar. \omega$ construction, where $\isovar$ is an
	\emph{iso-variable}.
  In general, we shall denote iso-variable by
  $\isovar_1, \isovar_2, \dots$ and we use the shorthands
  $\fix \overrightarrow{\isovar}$ or
  $\fix \isovar_1, \dots, \isovar_n$ and
  $\lambda \overrightarrow{\isovar}$ or
  $\lambda \isovar_1,\dots,\isovar_n$ for
  $\fix \isovar_1. \fix \isovar_2. \dots \fix \isovar_n .$ and
  $\lambda \isovar. \lambda \isovar_2. \dots \lambda \isovar_n$.
\end{itemize}

\medskip
\noindent
\textsf{\bfseries Convention.~} We write $(t_1, \dots, t_n)$ for
$\pv{t_1}{\pv{\dots}{t_n}}$ and $\bigoplus^n A$ (resp. $\bigotimes^n
A$) for $A\oplus\dots\oplus A$ (resp. $A\otimes\dots\otimes A$) $n$
times and $\omega_1 \dots \omega_n t$ for
a succession of {\lett}
constructions applying $\omega_n$ to $\omega_1$. We also consider
constructors to be right-associative, meaning that
$\fold{\inr{\pv{x}{y}}}$ should be read as
$\fold{(\inr{(\pv{x}{y})})}$. To avoid conflicts between
variables, we will always work up to $\alpha$-conversion and use
Barendregt's convention~\cite[p.26]{henk1984lambda}, which consists of
keeping the names of all bound and free variables distinct, even when this
remains implicit.
\begin{table}
	\[
		\begin{array}{@{}c@{}}
			\infer{\Psi;\emptyset\entaile \unit \colon \one}{}
			\qquad
			\infer{\Psi;x \colon A \entaile x \colon A}{}
			\qquad
			\infer{\Psi;\Delta\entaile \inl{t}\colon A \oplus B}{\Psi;\Delta\entail
			t\colon A}
			\qquad
			\infer{\Psi;\Delta\entaile\inr{t}\colon A \oplus B}{\Psi;\Delta\entaile t\colon B}
			\mynl
			\infer{
				\Psi;\Delta_1,\Delta_2;\entaile \pair{t_1}{t_2} \colon  A \otimes B
			}{
				\Psi;\Delta_1\entaile t_1 \colon  A
				&
				\Psi;\Delta_2\entaile t_2 \colon  B
			}
			\qquad
			\infer{\Psi; \Delta \entaile \fold{t} \colon  \mu X. A }{\Psi; \Delta
                                  \entaile t \colon  A[\mu X. A/X]}
			\qquad
			\infer{
				\Psi;\Delta\entaile \isoterm~t \colon  B
			}{
				\Psi\entailiso \isoterm \colon  A \iso B
				&
				\Psi;\Delta\entaile t \colon  A
			}
			\mynl
			\infer{\Psi;\Delta_1,\Delta_2\entaile \letv{(x_1, \dots,
			x_n)}{t_1}{t_2} \colon B}{\Psi;\Delta_1\entaile t_1 \colon  A_1 \otimes
			\dots \otimes A_n\qquad\Psi;\Delta_2, x_1 \colon  A_1, \dots, x_n
			\colon  A_n\entaile t_2 \colon  B}
		\end{array}
	\]
	\caption{Typing rules for terms.}
	\label{tab:typterm}
\end{table}
\begin{table}
	\[
		\begin{array}{c}
			\infer{\Psi, \isovar \colon T \entailiso \isovar \colon T}{}
			\qquad
			\infer{
				\Psi\entailiso \ffix \isovar.\omega \colon T
			}
			{
				\Psi, \isovar \colon T \entailiso \omega \colon T
			}
			\mynl
			\infer{
				\Psi\entailiso \omega_2~\omega_1 \colon T_2
			}
			{
				\Psi\entailiso \omega_1 \colon T_1
				&
				\Psi\entailiso \omega_2 \colon  T_1 \to T_2
			}
			\qquad
			\infer{
				\Psi\entailiso \lambda \isovar. \omega \colon  T_1 \to T_2
			}
			{
				\Psi, \isovar \colon T_1 \entailiso \omega \colon T_2
			}
			\mynl
			\infer{
				\Psi \entailiso
				\isobasique \colon  A \iso B.
			}{
				\begin{array}{l@{\quad}l@{\quad}l@{\qquad}l}
					\Psi;\Delta_1 \entaile v_1 \colon  A
					&
					\ldots
					&
					\Psi; \Delta_n\entaile v_n \colon  A
					&
					\forall i\not= j, v_i~\bot~v_j \\
					\Psi;\Delta_1\entaile e_1 \colon  B
					&
					\ldots
					& \Psi;\Delta_n\entaile e_n \colon  B & \forall i\not= j, e_i~\bot~ e_j
				\end{array}
			}
		\end{array}
	\]
    \caption{Typing rules for isos.}
    \label{tab:typisos}
\end{table}

\medskip
\noindent
\textsf{\bfseries Typing judgements.~}
Both base terms and isos feature their typing judgements, given
in~\Cref{tab:typterm} and~\Cref{tab:typisos}. Term typing judgements are of the
form $\Psi;\Delta\entaile t \colon A$ where $\Delta$ is a context of
term-variables of type $A$ and $\Psi$ is a context of iso-variables of type $T$
and isos typing judgements are of the form $\Psi\vdash_\omega \omega \colon T$.
While $\Delta$ is a \emph{linear} context, $\Psi$ is not, as an iso represents
a closed computation, and can be duplicated or erased at will. In the last rule
of~\Cref{tab:typisos}, the term variables in $\Delta$ are bound by the
pattern-matching construction: they are not visible outside of the term, thus
not appearing anymore in the typing context of the conclusion.

While~\cite{chardonnet2022curry} and~\cite{sabry2018symmetric} require
isos to be exhaustive (i.e. to cover all the possible values of their
input types) and non-overlapping (i.e. two clauses cannot match the same
value), we relax the exhaustivity requirement in this paper, in the
spirit of what was done in~\cite{nous21invcat}.
Non-overlapping is formalised by the notion of \emph{orthogonality}
between values, noted $v_1~\bot~v_2$.

\begin{definition}[Orthogonality]
  \label{def:orthogonality}
  We introduce a binary relation $\bot$ on terms. Given two terms $t_1, t_2$,
  $t_1~\bot~t_2$ holds if it can be derived inductively with the rules below;
  we say that $t_1$ and $t_2$ are orthogonal. The relation $\bot$ is defined as
  the smallest relation such that:
  \[
    \begin{array}{c}
      \infer{\inl{t_1}~\bot~\inr{t_2}}{}
      \qquad
      \infer{\inr{t_1}~\bot~\inl{t_2}}{}
      \qquad
      \infer{C_\bot[t_1]~\bot~C_\bot[t_2]}{t_1~\bot~t_2},
    \end{array}
  \]
  where the contexts $C_\bot$ are defined using the following grammar:
  \[C_\bot ::= [-] \alt \inl C_\bot \alt \inr C_\bot \alt
  \pv{C_\bot}{t} \alt \pv{t}{C_\bot} \alt \fold C_\bot \alt
  \letv{p}{t}{C_\bot}\]
\end{definition}

\begin{table}[t]
	  \[
    \begin{array}{c}
      \infer{
        \ffix \phi . \omega \to \omega[\ffix \phi . \omega / \phi]
      }{}
      \qquad
      \infer{
        (\lambda \phi . \omega_1) \omega_2 \to \omega_1[\omega_2/\phi]
      }{}
      \qquad
      \infer{
        \omega_1\omega_2 \to \omega'_1 \omega_2
      }{
        \omega_1 \to \omega'_1
      }
      \qquad
      \infer{\omega~t \to \omega'~t}{\omega \to \omega'}
    \end{array}
  \]
  \[
    \begin{array}{c}
      \infer{ \isobasique~v' \to \sigma(e_i)}{
        \match{\sigma}{v_i}{v'}}
      \qquad
      \infer{C_\to[t_1] \to C_\to[t_2]}{t_1 \to t_2}
      \qquad
      \infer{\letv{p}{v}{t} \to \sigma(t)}{\match{\sigma}{p}{v}}
      \qquad
    \end{array}
  \]
\caption{Evaluation relation $\to$.}
\label{def:rewriting-system}
\end{table}

\medskip
\noindent
\textsf{\bfseries Operational semantics.~}  The language comes
equipped with a rewriting system $\to$ on terms, defined
in~\Cref{def:rewriting-system}. As usual, we write $\to^*$ for the
reflexive transitive closure of $\to$.  The evaluation contexts
$C_\to$ are defined by the grammar $[~] \alt \inl{C_\to} \alt
\inr{C_\to} \alt \isobasique~C_\to \alt \letv{p}{C_\to}{t} \alt
\pv{C_\to}{v} \alt \pv{v}{C_\to} \alt C_\to~t \alt \fold{C_\to}$.
Note how the rewriting system follows a \emph{call-by-value} strategy
on terms and values, requiring that the argument of an iso be fully
evaluated to a value before firing the substitution. On the contrary,
we follow a \emph{call-by-name} strategy to simplify the manipulation
of the fixpoint. Note that unlike~\cite{chardonnet2022curry,
  sabry2018symmetric}, we do not require any form of termination and
isos are not required to be exhaustive: the rewriting system can
diverge or be stuck.
The evaluation of an iso applied to a value is dealt with by pattern-matching:
the input value will try to match one of the values from the clauses and
potentially create a substitution if the two values match, giving the
corresponding expression as an output under that substitution.
A substitution $\sigma$ is a mapping from a set of variables to terms.
The substitution of $\sigma$ on an expression $t$, written
$\sigma(t)$, is defined in the usual way by $\sigma(\unit) = \unit$;
$\sigma(x) = v$ if $\{x\mapsto v\}\subseteq \sigma$;
$\sigma(\inr{(t)}) = \inr{(\sigma(t))}$; $\sigma(\inl{(t)}) =
\inl{(\sigma(t))}$; $\sigma(\fold{(t)}) = \fold{(\sigma(t))}$
$\sigma(\pv{t}{t'}) = \pv{\sigma(t)}{\sigma(t')}$; $\sigma(\omega~t) =
\omega~\sigma(t)$ and $\sigma(\letv{p}{t_1}{t_2}) =
(\letv{p}{\sigma(t_1)}{\sigma(t_2)})$.
The support of a substitution, written
$\text{supp}(\sigma)$, is defined as $\set{x \mid (x\mapsto v) \in \sigma}$.

\begin{lemma}[Subject Reduction]
  \label{lemma:subject_reduction} If
  $\Psi;\Delta\entaile t \colon A$ and $t\rightarrow t'$, then
  $\Psi;\Delta\entaile t' \colon A$.\qed
\end{lemma}

The proof is similar to what has been done in~\cite{chardonnet2022curry}. As
the rewriting system is deterministic, confluence is direct; meanwhile, as we
are concerned with partial functions, progress is not guaranteed: a term can be
stuck, for example, $\set{\inl{(x)} \iso e}~\inr{(v)}$ does not reduce.

\medskip
\noindent
\textsf{\bfseries Inversion.~}
Finally, any iso $\omega \colon T$ can be inverted into an iso
$\omega^{-1} \colon T^{-1}$, such that their composition makes up the
identity. Intuitively, if $\omega$ is of type $A\iso B$, then $\omega^{-1}$
will be of type $B\iso A$. Inversion is defined as follows.
Given an iso-type $T$, we define its inverse $T^{-1}$ as: $(A\iso B)^{-1} =
B\iso A$ and $(T_1 \to T_2)^{-1} = T_1^{-1} \to T_2^{-1}$.
Given an iso $\omega$, we define its dual $\omega^{-1}$ as: $\isovar^{-1} =
\isovar$; $(\fix \isovar. \omega)^{-1} = \fix \isovar. \omega^{-1}$;
$(\omega_1~\omega_2)^{-1} = (\omega_1)^{-1} (\omega_2)^{-1}$; $(\lambda
\isovar. \omega)^{-1} = \lambda\isovar. (\omega)^{-1}$ and $\{(v_i\iso
e_i)_{i\in I}\}^{-1} = \{((v_i\iso e_i)^{-1})_{i\in I}\}$ and
\[\left(
    \begin{array}{l@{~~~}c@{~~~}l} v_1&{\iso}&\letv{p_1}{\omega_1~p_1'}{} \\
                                  && \cdots \\
                                  && \letv{p_n}{\isoterm_n~p'_n}{v'_1}
    \end{array}
    \right)^{-1} := \left(
      \begin{array}{l@{~~~}c@{~~~}l} v'_1&{\iso}&\letv{p'_n}{\isoterm_n^{-1}~p_n}{} \\
                                   && \cdots \\
                                   && \letv{p'_1}{\isoterm_1^{-1}~p_1}{v_1}
    \end{array}
    \right).
\]

\begin{property}[Inversion is an involution]
  For any well-typed iso $\omega$, we have $(\omega^{-1})^{-1} = \omega$.
\end{property}
\begin{proof}
  By a straightforward induction on $\omega$, notice that if $\omega =
  \isobasique$ then by definition we swap twice the order of the
  $\lett$ construction, hence recovering the original term.
\end{proof}

\begin{lemma}[Inversion is well-typed]
  \label{lem:inv-type}
  If $\isovar_1 \colon A_1\iso B_1\ldots  \isovar_n \colon A_n\iso B_n \vdash_\omega \omega \colon T$, then
  $\isovar_1 \colon B_1\iso A_1\ldots  \isovar_n \colon B_n\iso A_n  \vdash_\omega\omega^{-1} \colon T^{-1}$.
  \qed
\end{lemma}

\begin{lemma}[Inversion is preserved by evaluation]
  \label{lem:inv-rewriting}
  If $\omega \to \omega'$ then $\omega^{-1} \to \omega'^{-1}$.\qed
\end{lemma}

\begin{theorem}[Semantics of isos and their inversions~\cite{chardonnet2022curry}]
  \label{thm:isos-iso}
  For all well-typed isos $\vdash_\omega \omega \colon A\iso B$, and for all
  well-typed values $\entaile v \colon A$, if~ $(\omega~(\omega^{-1}~v))\rightarrow^*
  v'$ then $v = v'$.\qed
\end{theorem}

\begin{example}
  \label{ex:isos:map}
  Remember that $[A] = \mu X. \one\oplus (A\otimes X)$.  One can
  define the \emph{map} operator on lists with an iso of type
  $(A\iso B)\to [A]\iso [B]$, defined as
  \[
    \lambda \isolambdavar. \fix \isovar. \left\{
        [~] ~ {\iso} ~ [~] ~~|~~
        h :: t ~  {\iso} ~ \letv{h'}{\isolambdavar~h}{} \letv{t'}{\isovar~t}{}h' :: t'
    \right\},
  \]
  with the terms $[~] = \fold{(\inl{(\unit)})}$, representing the empty
  list, while the head and tail of the list is represented with $h::t = \fold{(\inr{(\pv{h}{t})})}$.
  Its inverse $\opn{map}^{-1}$ is
  \[
    \lambda\isolambdavar. \fix \isovar. \left\{
        [~] ~ {\iso} ~ [~] ~~|~~
        h' :: t' ~ {\iso} ~ \letv{t}{\isovar~t'}{} \letv{h}{\isolambdavar~h'}{} h :: t
    \right\}.
\]
Note that in the latter, the variable $\isolambdavar$ has type $B\iso A$.
If we consider the inverse of the term $(\opn{map}~\omega)$ we would obtain the term
$(\opn{map}^{-1}~\omega^{-1})$ where $\omega^{-1}$ would be of type $B\iso A$.
\end{example}

\begin{example}[Cantor Pairing]
  \label{ex:cantor}
  One can encode the Cantor Pairing between $\natS \otimes \natS \iso \natS$.
  First recall that the type of natural number $\natT$ is given by $\mu X.
  \one\oplus X$, then define $\overline{n}$ as the encoding of natural numbers
  into a closed value of type $\natT$ as $\overline{0} = \fold{(\inl{\unit})}$ and
  given a variable $x$ of type $\natT$, its successor is $\overline{S(x)} =
  \fold{(\inr{(x)})}$.  Omitting the $\overline{\ \cdot\ }$ operator for
  readability, the pairing is then defined as:
  \[\begin{array}{ll} \omega_1 \colon \natT \otimes \natT \iso (\natT \otimes \natT)
    \oplus \one  \\
    = \left\{\begin{array}{@{}lcl@{}}
      \pv{S(i)}{j} & {\iso} & \inl{(\pv{i}{S(j)})} \\
      \pv{0}{S(S(j))} & {\iso} & \inl{(\pv{S(j)}{0})} \\
      \pv{0}{S(0)} & {\iso} & \inl{(\pv{0}{0})} \\
      \pv{0}{0} & {\iso} & \inr{(\unit)} \end{array}\right\},
    \end{array}
    \hfill
    \begin{array}{@{}ll@{}}
      \omega_2 \colon (\natT \otimes \natT) \oplus \one \iso \natT  \\
      = \left\{\begin{array}{lcl}
        \inl{(x)} & {\iso} & \letv{y}{\isovar~x}{S(y)} \\
        \inr{(\unit)} & {\iso} & 0
    \end{array}\right\},
		\\

    \opn{CantorPairing} \colon \natT\otimes \natT \iso \natT \\
    = \fix \isovar. \left\{\begin{array}{@{}lcl@{}}
      x & {\iso} & \letv{y}{\omega_1~x}{} \\
      && \letv{z}{\omega_2~y}{z}
  \end{array}\right\},
  \end{array}\]
  where the variable $\isovar$ in $\omega_2$ is the one being bound
  by the ${\fix}$of the $\opn{CantorPairing}$ iso. Intuitively,
  $\omega_1$ realises one step of the Cantor Pairing evaluation while
  $\omega_2$ checks if we reached the end of the computation and either
  applies a recursive call, or stops.

  For instance, $\opn{CantorPairing}~\pv{1}{1}$ will match with the
  first clause of $\omega_1$, evaluating into $\inl{\pv{0}{2}}$, and
  then, inside $\omega_2$ the reduction
  $\opn{CantorPairing}~\pv{0}{2}$ will be triggered through the
  recursive call, evaluating the second clause of $\omega_1$, reducing
  to $\inl{\pv{1}{0}}$, etc.
\end{example}


\section{Expressivity}
\label{sec:RTM}
This section is devoted to assessing the expressivity of the
language. To that end, we rely on Reversible Turing Machine
(RTM)~\cite{axelsen11rtm}. We describe how to encode an RTM as an iso,
and prove that the iso realises the string semantics of the RTM.

\subsection{Recovering duplication, erasure and manipulation of constants}
\label{sec:dup}
Although the language is linear and reversible, since closed values
are all finite, and one can build isos to encode notions of
duplication, erasure, and constant manipulation thanks to partiality.

\begin{definition}[Duplication]\rm We define $\dup_A^S$ the iso of
	type $A\iso A\otimes A$ which can duplicate any closed value of type
	$A$ by induction on $A$, where $S$ is a set of pairs of a
	type-variable $X$ and an iso-variable $\isovar$, such that for every
	free-type-variable $X \subseteq A$, there exists a unique pair $(X,
	\isovar) \in S$ for some $\isovar$.

	The iso is defined by induction on $A$:
	$\dup_\one^S = \{\unit \iso \pv{\unit}{\unit}\}$, and
	\begin{itemize}
		\item $\dup_{A\otimes B}^S = \left\{\begin{array}{lcl}
				\pv{x}{y} & \iso & \letv{\pv{x_1}{x_2}}{\dup_A^S~x}{}
				\letv{\pv{y_1}{y_2}}{\dup_B^S~y}{} \\
				&& \pv{\pv{x_1}{y_1}}{\pv{x_2}{y_2}}
		\end{array}\right\}$;

	\item $\dup_{A\oplus B}^S = \left\{\begin{array}{lcl}
			\inl{(x)} & \iso & \letv{\pv{x_1}{x_2}}{\dup_A^S~x}{}
			\pv{\inl{(x_1)}}{\inl{(x_2)}} \\[0.5em]
			\inr{(y)} & \iso & \letv{\pv{y_1}{y_2}}{\dup_B^S~y}{}
			\pv{\inr{(y_1)}}{\inr{(y_2)}}
	\end{array}\right\}$;

\item If $(X, \_) \not\in S$: $\dup_{\mu X. A}^S = \fix\isovar.
	\left\{\begin{array}{l@{~}c@{~}l}
		\fold{(x)}
		& \iso
		&
		\letv{\pv{x_1}{x_2}}{Dup_{A[\mu X. A/X]}^{S\cup \set{(X , \isovar)}}~x}{} \\
		&& \pv{\fold{(x_1)}}{\fold{(x_2)}}
	\end{array}\right\}$;

\item If $(X, \isovar) \in S$:
	$\dup_{\mu X. A}^S = \set{x \iso
	\letv{\pv{x_1}{x_2}}{\isovar~x}{\pv{x_1}{x_2}}}$.
	\end{itemize}
	Remember that bound variables are assumed distinct following Barendregt's
	convention, allow for the well-definition of the isos above.
\end{definition}

\begin{lemma}[Properties of Duplication]
  \label{lem:duplication-invariant}
  \label{lem:duplication-typed}
  Given a closed type $A$, then $\dup_A^\emptyset$ is well-defined,
  and the iso $\dup_A^\emptyset$ is well typed of type
  $A\iso A\otimes A$.\qed
\end{lemma}

\begin{lemma}[Semantics of Duplication]
  \label{lem:duplication-semantics}
  Given a closed type $A$ and a closed value $v$ of type $A$, then
  $\dup_A^\emptyset~v \to^* \pv{v_1}{v_2}$ and $v = v_1 = v_2$.
	\qed
\end{lemma}

\begin{definition}[Constant manipulation]\rm
  \label{def:constant}
  We define $\opn{erase}_v \colon A \otimes \Sigma^T \iso A$ which
  erases its second argument when its value is $v$ as $\{\pv{x}{v}
  \iso x\}$. Reversed, it turns any $x$ into $\pv{x}{v}$.
\end{definition}

\subsection{Definition of Reversible Turing Machine}

\begin{definition}[Reversible Turing Machine~\cite{axelsen11rtm}]\rm
  Let $M = (Q, \Sigma, \delta, b, q_s, q_f)$ be a Turing Machine, where
  $Q$ is a set of states, $\Sigma = \set{b, a_1, \dots, a_n}$ is a
  finite set of tape symbols (in the following, $a_i$ and $b$ always
  refer to elements of $\Sigma$),
  $\delta \subseteq \Delta = (Q\times
  [(\Sigma\times\Sigma)\cup\set{\leftarrow, \downarrow,
    \rightarrow}]\times Q)$ is a partial relation defining the
  transition relation such that there must be no transitions leading
  out of $q_f$ nor into $q_s$, $b$ a blank symbol and $q_s$ and $q_f$
  the initial and final states. We say that $M$ is a \emph{Reversible
    Turing Machine} (RTM) if it is:
  \begin{itemize}
      \item \textit{forward} deterministic: for any two distinct pairs
  of triples $(q_1, a_1, q_1')$ and $(q_2, a_2, q_2')$ in $\delta$, if
  $q_1 = q_2$ then $a_1 = (s_1, s_1')$ and $a_2 = (s_2, s_2')$ and
  $s_1\not= s_2$.

  \item \textit{Backward} deterministic: for any two distinct pairs
  of triples $(q_1, a_1, q_1')$ and $(q_2, a_2, q_2')$ in $\delta$, if
  $q_1' = q_2'$ then $a_1 = (s_1, s_1')$ and $a_2 = (s_2, s_2')$ and
  $s_1'\not= s_2'$.
\end{itemize}
\end{definition}

\begin{definition}[Configurations~\cite{axelsen11rtm}]
  \label{confTM}\rm
  A \emph{configuration} of a RTM is a tuple $(q, (l, s, r)) \in
  \opn{Conf} = Q\times (\Sigma^* \times \Sigma \times \Sigma^*)$ where
  $q$ is the internal state, $l, r$ are the left and right parts of
  the tape (as string) and $s\in \Sigma$ is the current symbol being
  scanned. A configuration is \emph{standard} when the cursor is on
  the immediate left of a finite, blank-free string $s \in
  (\Sigma\setminus\set{b})^*$ and the rest is blank, i.e. it is in
  configuration $(q, (\epsilon, b, s))$ for some $q$, where $\epsilon$
  is the empty string, representing an infinite sequence of blank
  symbols $b$.
\end{definition}

\begin{definition}[RTM Transition~\cite{axelsen11rtm}]\rm
  An RTM $M$ in configuration $C = (q, (l, s, r))$ goes to
  a configuration $C' = (q', (l', s', r'))$, written $T\vdash
  C\rightsquigarrow C'$ in a single step if there exists
  a transition $(q, a, q')\in\delta$ where $a$ is either $(s, s')$,
  and then $l=l'$ and $r=r'$ or $a\in\{\leftarrow,\downarrow,
  \rightarrow\}$, and we have for the case $a=\leftarrow$:
  $l' = l \cdot s$ and for $r = x \cdot r_2$ we have $s' = x$ and
  $r' = r_2$, similarly for the case $a = \rightarrow$ and for the
  case $a =\downarrow$ we have $l' = l$ and $r' = r$ and $s = s'$.
\end{definition}

The semantics of an RTM is given on \textit{standard configurations} of
the form $(q, (\epsilon, b, s))$ where $q$ is a state, $\epsilon$ is
the finite string standing for a blank-filled tape, and $s$ is the
blank-free, finite input of the RTM.

\begin{definition}[String Semantics~\cite{axelsen11rtm}]
  \label{RTM:String-Sem}\rm The semantics of a RTM $M$, written
  $\opn{Sem}(M)$ is defined on standards configurations and is given
  by the set $\opn{Sem}(M) = \set{(s, s') \in ((\Sigma\backslash\set{b})^*
  \times (\Sigma\backslash\set{b})^*) \mid M \vdash (q_s, (\epsilon,
  b, s)) \rightsquigarrow^* (q_f, (\epsilon, b, s'))}$.
\end{definition}

\begin{theorem}[Properties of RTM~\cite{axelsen11rtm}]
  For all RTM $M$, $\opn{Sem}(M)$ is the graph of an injective
  function. Conversely, all injective computable functions (on a
  tape) are realisable by a RTM. Finally, any Turing Machine can be
  simulated by a Reversible Turing Machine.\qed
\end{theorem}

\subsection{Encoding RTMs as Isos}
A RTM configuration is a set-based construction that we can model
using the type constructors available in our language. Because the
transition relation $\delta$ is backward and forward deterministic, it
can be encoded as an iso. Several issues need to be dealt with; we
discuss them in this section.

\smallskip
\noindent
\textsf{\bfseries Encoding configurations.~}
The set of states $Q = \{q_1, \dots, q_n\}$ is modeled with the type
$Q^T=\1\oplus\cdots\oplus\1$ ($n$ times). The encoding of the state
$q_i$ is then a closed value $q_i^T$. They are pairwise orthogonal.
The set $\Sigma$ of tape symbols is represented similarly by
$\Sigma^T=\1\oplus\cdots\oplus\1$, and the encoding of the tape symbol
$a$ is $a^T$.
We then define the type of configurations in the obvious manner: a
configuration $C = (q, (l, s, r))$ corresponds to a closed value
$\opn{isos}(C)$ of type
$Q^T \otimes ([\Sigma^T] \otimes \Sigma^T \otimes [\Sigma^T])$.

\begin{definition}[Encoding of Configurations]\rm We define the type
  of configurations as $C^T = (Q^T\otimes ([\Sigma^T] \otimes \Sigma
  \otimes [\Sigma^T]))$. Given a configuration $C = (q, ((\epsilon,
  a_1, \dots, a_n), a, (a_1', \dots, a_m', \epsilon)))$, it is encoded
  as $\opn{isos}(C) = (q^T, ([a_n^T, \dots, a_1^T], a^T, [a_1'^T,
  \dots, a_m'^T]))$. For example, the standard configuration $C =
  (q_s, (\epsilon, b, [a_1, \dots, a_n]))$ is represented as
  $\opn{isos}(C) = (q_s^T, ([], b^T, [a_1^T, \dots, a_n^T]))$.
\end{definition}

\smallskip
\noindent
\textsf{\bfseries Encoding the transition relation $\delta$.~}
A limitation of our language is that every sub-computation has to be reversible
and does not support infinite data structures such as streams. In the context
of RTMs, the empty string $\epsilon$ is identified with an infinite string of
blank symbols. If this can be formalised in set theory, in our limited model,
we cannot emit blank symbols out of thin air without caution.

In order to simulate an infinite amount of blank symbols on both sides
of the tape during the evaluation, we provide an iso that grows the
size of the two tapes on both ends by blank symbols at each transition
step.
The iso $\growth$ is shown in
Table~\ref{tab:useful-functions}. It is built using three auxiliary
functions, written in a Haskell-like notation.
$\len$ sends a closed value $[v_1, \dots, v_n]$ to
$\pair{[v_1, \dots, v_n]}{\ov{n}}$.
$\snocc$ sends $\pair{[v_1, \dots, v_n]}{v, \ov{n}}$ to
$\pair{[v_1, \dots, v_n, v]}{v, \ov{n}}$.
$\snoc$ sends $\pv{[v_1, \dots, v_n]}{v}$ to
$\pv{[v_1, \dots, v_n, v]}{v}$.
Finally, $\growth$ sends
$\pv{[a_1^T\!, \dots, a_n^T]}{[a_1'^T\!, \dots, a_m'^T]}$ to
$\pv{[a_1^T, \dots, a_n^T, b^T]}{[a_1'^T, \dots, a_m'^T, b^T]}$.

Now, given a RTM $M = (Q, \Sigma, \delta, b, q_s, q_f)$, a relation
$(q, r, q') \in \delta$ is encoded as a clause between values
$\opn{iso}(q,r,q') = v_1 \iso v_2$ of type $C^T \iso C^T$. These clauses
are defined by case analysis on $r$ as follows.
When $x, x', z, y$ and $y'$ are variables:
\begin{itemize}
\item $\opn{iso}(q, \rightarrow, q') =
  (q^T, (x', z, y :: y')) \iso \letv{(l,
    r)}{\growth~(x', y')}{(q'^T, (z :: l, y, r))}$,
\item $\opn{iso}(q, \leftarrow, q') =
  (q^T, (x::x', z, y')) \iso \letv{(l, r)}{\growth~(x',
    y')}{(q'^T, (l, x, z :: r))}$,
\item $\opn{iso}(q, \downarrow, q') =
  (q^T, (x', z, x')) \iso \letv{(l, r)}{\growth~(x',
    y')}{(q'^T, (l, z, r))}$,
\item $\opn{iso}(q, (s, s'), q') =
  (q^T, (x', s^T, y')) \iso \letv{(l, r)}{\growth~(x',
    y')}{(q'^T, (l, s'^T, r))}$.
\end{itemize}
The encoding of the RTM $M$ is then the iso $\opn{isos}(M)$ whose
clauses are the encoding of each rule of the transition relation
$\delta$, of type $\opn{Conf}^T\iso\opn{Conf}^T$.

\begin{table}
	\begin{tabular}{@{}l}
		$
		\begin{array}{|l}
			\len:[A]\iso[A] \otimes \natT\\
			\begin{array}{@{}l@{\,}l@{~}c@{~}l}
				\len & [~] & {\iso} & ([~], 0) \\
				\len & h :: t & {\iso} & \letv{(t', n)}{\len~t}{} \\
				&   & & (h :: t', S(n)) \\
			\end{array}
		\end{array}
		$
		\\[6ex]
		$
		\begin{array}{|l}
			\snocc \colon [A] \otimes A \otimes \natT \iso [A] \otimes A \otimes \natT \\
			\begin{array}{@{}l@{\,}l@{~}c@{~}l}
				\snocc & ([~], x, 0) & {\iso} & \letv{(x_1, x_2)}{\dup_A^\emptyset x}{}\\
				& & & ([x_1], x_2, 0) \\
				\snocc & (h::t, x, S(n)) & {\iso} & \letv{(t', x', n')}{\snocc (t, x, n)}{} \\
				&&&(h :: t', x', S(n'))
			\end{array}
		\end{array}
		$
		\\[8ex]
		$
		\begin{array}{|l}
			\snoc \colon [A] \otimes A\iso [A] \otimes A \\
			\begin{array}{@{}l@{\,}l@{~}c@{~}l}
				\snoc
				& (x, y) & {\iso} & \letv{(x', n)}{\len~x}{} \\
				& & & \letv{(x'', y', n')}{\snocc~(x', y, n)}{} \\
				& & & \letv{n''}{\set{x ~\iso~ S x}~n'}{} \\
				& & & \letv{z}{\len^{-1}~(x'', n'')}{} (z, y')
			\end{array}
		\end{array}
		$
		\\[8ex]
		$
		\begin{array}{|l}
			\growth \colon [\Sigma^T] \otimes [\Sigma^T] \iso [\Sigma^T]
			\otimes [\Sigma^T]\\
			\begin{array}{@{}l@{\,}l@{~}c@{~}l}
				\growth
				& (l, r)
				& {\iso}
				& \letv{\pv{l'}{b_1}}{\snoc\pv{l}{b^T}}{} \\
				&& & \letv{\pv{r'}{b_2}}{\snoc\pv{r}{b^T}}{} \\
				&& & \letv{l''}{\erase_b \pv{l'}{b_1}}{} \\
				&& & \letv{r''}{\erase_b \pv{r'}{b_2}}{} (l'', r'')
			\end{array}
		\end{array}
		$
		\\[8ex]
		$
		\begin{array}{|l}
			\It: (A\iso A \otimes (\one\oplus\one)) \to (A\iso
			A\otimes \natT)\\
			\begin{array}{@{}l@{\,}l@{~}c@{~}l}
				\It \isolambdavar
				& x& \iso & \letv{y}{\isolambdavar~x}{}\\
				&&& \letv{z}{
					\left\{\begin{array}{l@{~}c@{~}l}
						(y, \tc)
						&\iso
						& \letv{(z,n)}{(\It\,\isolambdavar)~y}{(z, S~n)} \\
						(y, \fc) & \iso & (y, 0)
					\end{array}\right\}
					~y}{}
				z
			\end{array}
		\end{array}
		$
		\\[6ex]
		$
		\begin{array}{|l}
			\rmBlank \colon [\Sigma] \iso [\Sigma] \otimes \natS
			\\
			\begin{array}{@{}l@{\,}l@{~}c@{~}l}
				\rmBlank & [] & {\iso} & ([], 0) \\
				\rmBlank & b^T :: t & {\iso} & \letv{(t', n)}{\rmBlank~t}{(t', S(n))} \\
				\rmBlank & a_1^T :: t & {\iso} & ((a_1^T :: t), 0) \\
				\vdots& \vdots & \vdots & \vdots \\
				\rmBlank & a_n^T :: t & {\iso} & ((a_n^T :: t), 0)
			\end{array}
		\end{array}
		$
		\\[10ex]
		$
		\begin{array}{|l}
			\rev_{\text{aux}} \colon [A] \otimes [A] \iso [A]
			\otimes [A]
			\\
			\begin{array}{@{}l@{\,}l@{~}c@{~}l}
				\rev_{\text{aux}} & ([], y) & {\iso} & ([], y) \\
				\rev_{\text{aux}} & (h :: t, y) & {\iso} & \letv{(h_1, h_2)}{\dup_A^\emptyset~h}{} \\
				&& & \letv{(t_1, t_2)}{\isovar (t, h_2 :: y)}{} \\
				&& & (h_1 :: t_1, t_2)
			\end{array}
		\end{array}
		$
		\\[8ex]
		$
		\begin{array}{|l}
			\rev \colon [A] \iso [A] \otimes [A]\\
			\rev = \set{x \iso \letv{(t_1,t_2)}{\rev_{\text{aux}}~(x, [])}{(t_1, t_2)}}
		\end{array}
		$
		\\[4ex]
		$
		\begin{array}{|l}
			\cleanUp: C^T \otimes \natT \iso C^T \otimes
			\natT \otimes \natT \otimes \natT \otimes [\Sigma^T]
			\\
			\begin{array}{@{}l@{\,}l@{~}c@{~}l}
				\cleanUp
				&
				((x, (l, y, r)), n)
				& {\iso} &
				\letv{(l', n_1)}{\rmBlank~l}{}\\
				&&&\letv{(r_{\text{ori}}, r_{\text{rev}})}{\rev~r}{}\\
				&&&\letv{(r', n_2)}{\rmBlank~r_{\rev}}{}\\
				&&&((x, (l', y, r')), n, n_1, n_2, r_{\text{ori}})
			\end{array}
		\end{array}
		$
	\end{tabular}
	\caption{Some useful isos for the encoding.}
	\label{tab:useful-functions}
\end{table}

\smallskip
\noindent
\textsf{\bfseries Encoding successive applications of $\delta$.~}
The transition $\delta$ needs to be iterated until the final state is
reached. This behavior can be emulated in our language using the iso
$\It$, defined in Table~\ref{tab:useful-functions}.  The iso
$\It\,\omega$ is typed with $(A\iso A\otimes \natT)$. Fed with a
value of type $A$, it iterates $\omega$ until $\fc$ is met. It then
returns the result together with the number of iterations.

To iterate $\opn{iso}(M)$, we then only need to modify $\opn{iso}$ to
return a boolean stating whether $q_f$ was met. This can be done straightforwardly, yielding an iso $\opn{isos}_\boolT(M))$ of
type
$
  \opn{Conf}^T\iso \opn{Conf}^T\otimes(\1\oplus\1).
$
With such an iso, given $M$ be a RTM such that $M \vdash (q_s, (\epsilon, b,
s)) \rightsquigarrow^{n+1} (q_f, (\epsilon, b, (a_1, \dots, a_n)))$, then
$\It(\opn{isos}_\boolT(M))~(q_s^T, ([b^T], b^T, s^T))$ reduces to the
encoding term $((q_f^T, ([b^T, \dots, b^T], b^T, [a_1^T, \dots, a_n^T, b^T,
\dots, b^T])), \overline{n})$. If it were not for the additional blank tape
elements, we would have the encoding of the final configuration.

\smallskip
\noindent
\textsf{\bfseries Recovering a canonical presentation.~}
Removing blank states at the \emph{beginning} of a list is easy: for
instance, it can be done with the iso $\opn{rmBlank}$, shown
in Table~\ref{tab:useful-functions}. Cleaning up the tail of the list
can then be done by reverting the list, using, e.g. $\opn{rev}$ in the
same table. By abuse of notation, we use constants in some patterns:
an exact representation would use \Cref{def:constant}. Finally, we can
define the operator $\cleanUp$, solving the issue raised in the
previous paragraph.
In particular, given a RTM $M$ and an initial configuration $C$ such
that $M \vdash C \rightsquigarrow C' = (q, (\epsilon, b, (a_1, \dots,
a_n)))$, then we have that $\cleanUp~(\It(\opn{isos}_\boolT(M)) C^T)
\to^* ((q^T, ([], b^T, [a_1^T, \dots, a_n^T])), v)$, where $v$ is of
type $\natT \otimes \natT \otimes \natT \otimes [\Sigma^T]$. If we
want to claim that we indeed capture the operational behaviour of
RTMS, we need to get rid of this value $v$.

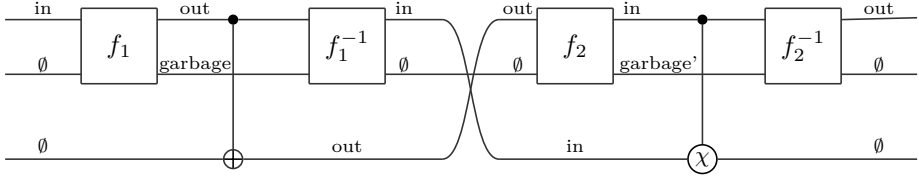
\begin{figure}
  \begin{tikzpicture}
	\begin{pgfonlayer}{nodelayer}
		\node [style=none] (0) at (-2, 0) {};
		\node [style=none] (1) at (-1, 0) {};
		\node [style=none] (2) at (-1, -1) {};
		\node [style=none] (3) at (-2, -1) {};
		\node [style=none] (4) at (1, 0) {};
		\node [style=none] (5) at (1, -1) {};
		\node [style=none] (6) at (2, -1) {};
		\node [style=none] (7) at (2, 0) {};
		\node [style=none] (8) at (4, 0) {};
		\node [style=none] (9) at (4, -1) {};
		\node [style=none] (10) at (5, -1) {};
		\node [style=none] (11) at (5, 0) {};
		\node [style=none] (12) at (7, 0) {};
		\node [style=none] (13) at (7, -1) {};
		\node [style=none] (14) at (8, -1) {};
		\node [style=none] (15) at (8, 0) {};
		\node [style=none] (16) at (-1.5, -0.5) {$f_1$};
		\node [style=none] (17) at (1.5, -0.5) {$f_1^{-1}$};
		\node [style=none] (18) at (4.5, -0.5) {$f_2$};
		\node [style=none] (19) at (7.5, -0.5) {$f_2^{-1}$};
		\node [style=none] (20) at (-3, -2) {};
		\node [style=none] (21) at (-3, -0.875) {};
		\node [style=none] (22) at (-3, -0.15) {};
		\node [style=dot] (23) at (0, -0.15) {};
		\node [style=plus] (24) at (0, -2) {};
		\node [style=none] (25) at (2.75, -0.15) {};
		\node [style=none] (26) at (3.5, -0.15) {};
		\node [style=none] (27) at (2.75, -2) {};
		\node [style=none] (28) at (3.5, -2) {};
		\node [style=dot] (29) at (6.175, -0.15) {};
		\node [style=white dot] (30) at (6.175, -2) {$\chi$};
		\node [style=none] (31) at (9, -2) {};
		\node [style=none] (32) at (9, -0.875) {};
		\node [style=none] (33) at (9, -0.125) {};
		\node [style=none] (35) at (2, -0.875) {};
		\node [style=none] (37) at (8, -0.875) {};
		\node [style=none] (38) at (8, -0.15) {};
		\node [style=none] (39) at (7, -0.875) {};
		\node [style=none] (40) at (7, -0.15) {};
		\node [style=none] (41) at (5, -0.875) {};
		\node [style=none] (42) at (5, -0.15) {};
		\node [style=none] (43) at (4, -0.875) {};
		\node [style=none] (44) at (4, -0.15) {};
		\node [style=none] (45) at (-2, -0.875) {};
		\node [style=none] (46) at (-2, -0.15) {};
		\node [style=none] (47) at (-1, -0.875) {};
		\node [style=none] (48) at (-1, -0.15) {};
		\node [style=none] (49) at (1, -0.875) {};
		\node [style=none] (50) at (1, -0.15) {};
		\node [style=none] (52) at (2, -0.15) {};
		\node [style=none, font={\scriptsize}] (53) at (-2.5, 0) {in};
		\node [style=none, font={\scriptsize}] (54) at (-2.5, -0.75) {$\emptyset$};
		\node [style=none, font={\scriptsize}] (55) at (-2.5, -1.825) {$\emptyset$};
		\node [style=none, font={\scriptsize}] (56) at (-0.5, 0) {out};
		\node [style=none, font={\scriptsize}] (57) at (-0.5, -0.75) {garbage};
		\node [style=none, font={\scriptsize}] (58) at (1.5, -1.825) {out};
		\node [style=none, font={\scriptsize}] (59) at (2.25, 0) {in};
		\node [style=none, font={\scriptsize}] (60) at (2.25, -0.75) {$\emptyset$};
		\node [style=none, font={\scriptsize}] (61) at (5.25, 0) {in};
		\node [style=none, font={\scriptsize}] (62) at (5.6, -0.75) {garbage'};
		\node [style=none, font={\scriptsize}] (63) at (3.75, 0) {out};
		\node [style=none, font={\scriptsize}] (64) at (3.75, -0.75) {$\emptyset$};
		\node [style=none, font={\scriptsize}] (65) at (4.5, -1.825) {in};
		\node [style=none, font={\scriptsize}] (66) at (8.5, -1.825) {$\emptyset$};
		\node [style=none, font={\scriptsize}] (67) at (8.5, 0) {out};
		\node [style=none, font={\scriptsize}] (68) at (8.5, -0.75) {$\emptyset$};
	\end{pgfonlayer}
	\begin{pgfonlayer}{edgelayer}
		\draw (3.center) to (0.center);
		\draw (0.center) to (1.center);
		\draw (1.center) to (2.center);
		\draw (2.center) to (3.center);
		\draw (5.center) to (4.center);
		\draw (4.center) to (7.center);
		\draw (7.center) to (6.center);
		\draw (6.center) to (5.center);
		\draw (9.center) to (8.center);
		\draw (8.center) to (11.center);
		\draw (11.center) to (10.center);
		\draw (10.center) to (9.center);
		\draw (13.center) to (12.center);
		\draw (12.center) to (15.center);
		\draw (15.center) to (14.center);
		\draw (14.center) to (13.center);
		\draw (31.center) to (30);
		\draw (27.center) to (24);
		\draw (20.center) to (24);
		\draw (23) to (24);
		\draw (28.center) to (30);
		\draw (22.center) to (46.center);
		\draw (21.center) to (45.center);
		\draw (47.center) to (49.center);
		\draw (48.center) to (23);
		\draw (23) to (50.center);
		\draw (52.center) to (25.center);
		\draw (26.center) to (44.center);
		\draw (35.center) to (43.center);
		\draw (42.center) to (29);
		\draw (29) to (40.center);
		\draw (38.center) to (33.center);
		\draw (37.center) to (32.center);
		\draw (41.center) to (39.center);
		\draw [in=0, out=180, looseness=0.50] (26.center) to (27.center);
		\draw [in=180, out=0, looseness=0.50] (25.center) to (28.center);
		\draw (29) to (30);
	\end{pgfonlayer}
\end{tikzpicture}
  \caption{Reversibly removing additional garbage from some process.}
  \label{fig:garbage-removal}
\end{figure}

\smallskip
\noindent
\textsf{\bfseries Getting rid of the garbage.~}
To discard this value $v$, we rely on Bennett's
trick~\cite{bennett1973logical}, shown in \Cref{fig:garbage-removal}.
Given two Turing machines $f_1$ and $f_2$ and some input $\mathtt{in}$
such that if $f_1(\mathtt{in}) = \mathtt{out} \otimes \mathtt{garbage}$ and
$f_2(\mathtt{out}) = \mathtt{in} \otimes \mathtt{garbage'}$, then the process
consists of taking additional tapes in the Turing Machine in order to
reversibly duplicate (represented by the $\oplus$) or reversibly erase
some data (represented by the $\chi$) in order to recover only the
output of $f_1$, without any garbage.

Given an iso $\omega \colon A\iso B\otimes C$ and
$\omega' \colon B \iso A \otimes C'$ where $C, C'$ represent garbage, we
can build an iso from $A\iso B$ as follows, where the variables
$x, y, z$ (and their indices) respectively correspond to the first,
second, and third wire of~\Cref{fig:garbage-removal}. This operator
makes use of the iso \opn{Dup} discussed in Section~\ref{sec:dup}.
\[
\begin{array}{lcl}
  \garRem{\omega}{\omega'}~x_1
  & \iso
  & \letv{\pv{x_2}{y}}{\omega~x_1}{}
    \letv{\pv{x_3}{z}}{\dup_B^\emptyset~x_2}{}\\
  && \letv{x_4}{\omega^{-1}~\pv{x_3}{y}}{}
   \letv{\pv{z_2}{y_2}}{\omega'~z}{}\\
  && \letv{z_3}{(\dup_B^\emptyset)^{-1}~\pv{z_2}{x_4}}{}
   \letv{z_4}{\omega'^{-1}~\pv{z_3}{y_2}}{}
             z_4.
\end{array}
\]

\begin{theorem}[Capturing the exact semantics of a RTM]
  For all RTM $M$ with standard configurations $C = (q_s,
  (\epsilon, b, s))$ and $C' = (q_f, (\epsilon, b, s'))$ such that
  $M \vdash C \rightsquigarrow^* C'$, we have
  \[\garRem{\cleanUp(\It
      (\opn{isos}_\boolT(M)))}{\cleanUp(\It
      (\opn{isos}_\boolT(M^{-1})))}~\opn{isos}(C) \to^* \opn{isos}(C')\]
  The behavior of RTMs is thus captured by the language.\qed
\end{theorem}

\section{Categorical Background}

We aim at providing a denotational semantics for the programming language
introduced above, meaning a mathematical interpretation abstract to the syntax.
Our approach is categorical, in the spirit of many others before us. Programs
are compositional by design, making it natural to interpret in a
framework ruled by compositionality.  Types are usually interpreted as objects
in a category $\CC$, and terms as morphisms in this category. We have seen that
the main feature of our programming language is reversibility and its terms can
be seen as partial isomorphisms, or partial injections.  We want this property
to be carried on the interpretation, and we present in this section the proper
categories to do so.  The category of sets and partial injective functions,
written $\PInj$, will be the recurring example throughout this section to help
the intuition.

\subsection{Join inverse rig category}

The axiomatisation of join inverse rig categories gives the conditions
for the morphisms of a category to be \emph{partial injections}.
First, the notion of restriction allows to capture the \emph{actual}
domain of a morphism through a partial identity function.
Historically, \emph{inverse} categories \cite{kastl1979inverse} were
introduced before \emph{restriction} categories, but the latter are
more convenient to introduce the subject.

\begin{definition}[Restriction~\cite{cockett2002restriction-I}]
	\label{def:restr}\rm A restriction structure is an operator that
        maps each morphism $f:A\rightarrow B$ to a morphism $\res
        f:A\rightarrow A$ such that for all $g$ and $h$ such that the
        domain of $g$ is $A$ and the domain of $h$ is $B$ we have
        $f \circ \res f = f$, $\res f \circ \res g = \res g \circ \res
        f$, $\res{f \circ \res g} = \res f \circ \res g$ and $\res h
        \circ f = f \circ \res{h \circ f}$.
  A morphism $f$ is said to be \emph{total} if $\res f = 1_A$.
  A category with a restriction structure is called a
  \emph{restriction category}.
  A functor $F:\mathcal C \rightarrow \mathcal D$ is a
  \emph{restriction functor} if $\res{F(f)} = F\left(\res f\right)$ for all
  morphism $f$ of $\mathcal C$. The definition is canonically
  extended to bifunctors.
  When unambiguous, we write $gf$ for the
composition $g\circ f$.
\end{definition}

\begin{example}
  Given sets $A,B$ and a partial function $f:A\to B$ defined on
  $A'\subseteq A$ and undefined on $A\setminus A'$, the restriction of
  $f$ is $\res f:A\to A$, the identity on $A'\subseteq A$ and
  undefined on $A\setminus A'$. This example shows that $\PInj$ is a
  restriction category.
\end{example}

To interpret reversibility, we need to introduce a notion of
reversed process, a process that exactly reverses another process.
This is given by a generalised notion of inverse.

\begin{definition}[Inverse category~\cite{kaarsgaard2017join}]
  \label{def:invcat}\rm
  An \emph{inverse category} is a restriction category where all morphisms
  are partial isomorphisms; meaning that for $f:A\rightarrow B$, there
  exists a unique $f^{\circ} \colon B\rightarrow A$ such that
  $f^{\circ} \circ f = \res f$ and $f \circ f^{\circ} = \rc f$.
\end{definition}

\begin{example}
	In $\PInj$, let us consider the partial function $f\colon\{0,1\}\to\{0,1\}$
	as $f(0)=1$ and undefined on $1$. Its restriction $\res f$
	is undefined on $1$ also but $\res f(0)=0$. Its \emph{inverse}
	$f^\circ$ is undefined on $0$ and such that $f^\circ(1)=0$.
\end{example}

The example above generalises and $\PInj$ is an actual inverse category. Even
more, it is \emph{the} inverse category: \cite{kastl1979inverse} proves that
every locally small inverse category is isomorphic to a subcategory of $\PInj$.

\begin{definition}[Restriction compatible~\cite{kaarsgaard2017join}]
  \label{def:restcomp}\rm
  Two morphisms $f,g:A\to B$ in a restriction category $\mathcal{C}$
  are restriction compatible if $f\res g = g\res f$.  The relation is
  written $f \smile g$.  If $\mathcal{C}$ is an inverse category, they
  are inverse compatible if $f\smile g$ and
  $f^{\circ} \smile g^{\circ}$, noted $f\asymp g$.
  A set $S$ of morphisms of the same type $A\to B$ is restriction compatible
  (\textit{resp.}  inverse compatible) if all elements of $S$ are
  pairwise restriction compatible (\textit{resp.} inverse compatible).
\end{definition}

\begin{definition}[Partial order~\cite{cockett2002restriction-I}]\label{def:order}\rm
  Let $f,g:A\to B$ be two morphisms in a restriction category. We
  then define $f \leq g$ as $g\res f = f$.
\end{definition}

\begin{definition}[Joins~\cite{guo2012products}]\label{def:join}\rm
  A restriction category $\mathcal{C}$ is equipped with joins if for
  all restriction compatible sets $S$ of morphisms $A\to B$, there
  exists
  $\bigvee_{s\in S} s:A\to B$ morphism of $\mathcal{C}$ such
  that, whenever $t \colon A \to B$ and whenever for all $s\in S$, $s\leq t$,
  $s \leq \bigvee_{s\in S} s$,
  $\bigvee_{s\in S} s \leq t$,
  $\res{\bigvee_{s\in S} s} = \bigvee_{s\in S} \res s$,
  $f\circ\left(\bigvee_{s\in S} s\right)
  = \bigvee_{s\in S} fs$,
  $\left(\bigvee_{s\in
      S} s\right)\circ g
    =\bigvee_{s\in S} sg$.
  Such a category is called a \emph{join restriction category}. An
  inverse category with joins is called a \emph{join inverse
  category}.
\end{definition}

Building up from Definition~\ref{def:restr}, a \emph{join restriction
functor} is a restriction functor that preserves all thus constructed joins.

\begin{definition}[Zero~\cite{kaarsgaard2017join}]\label{def:zero}\rm
  Since $\emptyset \subseteq \text{Hom}_{\mathcal C} (A,B)$, and since
  all of its elements are restriction compatible, there exists a
  morphism $0_{A,B} \doteq\bigvee_{s\in\emptyset} s$, called
  \emph{zero}.
  It satisfies the following equations:
  $f0=0$,
  $0g=0$,
  $0_{A,B}^{\circ} = 0_{B,A}$,
  $\res{0_{A,B}} = 0_{A,A}$.
\end{definition}

\begin{definition}[Restriction Zero]
	A restriction category $\mathcal C$ has a \emph{restriction zero} object 0
	iff for all objects $A$ and $B$, there exists a unique morphism $0_{A,B}
	\colon A \to B$ that factors through $0$ and satisfies $\res{0_{A,B}}
	=0_{A,A}$.
\end{definition}

\begin{definition}[Disjointness tensor~\cite{giles2014investigation}]
  \label{def:disten}\rm
	An inverse category $\mathcal{C}$ is said to have a \textit{disjointness
	tensor} if it is equipped with a symmetric monoidal restriction bifunctor
	$.\oplus .:\mathcal{C}\times\mathcal{C}\rightarrow\mathcal{C}$, with as unit
	a restriction zero $0$ and morphisms $\iota_l:A\rightarrow A\oplus B$ and
	$\iota_r:B\rightarrow A\oplus B$ that are total, jointly epic, and such that
	their inverses are jointly monic and $\rc{\iota_l}~\rc{\iota_r} = 0_{A\oplus
	B}$.
\end{definition}

\begin{definition}[\cite{kaarsgaard2021join}]\label{def:rig}\rm
	Let us consider a join inverse category equipped with a symmetric monoidal
	tensor product $(\otimes,1)$ and a disjointness tensor $(\oplus,0)$ that are
	join preserving, and such that there are isomorphisms $\delta_{A,B,C} \colon A
	\otimes (B\oplus C) \rightarrow (A\otimes B)\oplus (A\otimes C)$ and $\nu_A
	\colon A \otimes 0 \rightarrow 0$. This is called a \emph{join inverse rig
	category}.
\end{definition}

\subsection{\texorpdfstring{$\DCPO$}{DCPO}-category}

We use the vocabulary of enriched category theory to shorten the discussion in
this section. The notions of enrichment required to understand the semantics
later is basic and should not frighten the reader. Categories in computer
science are usually \emph{locally small}, meaning that given two objects $A$ and
$B$, there is a \emph{set} of morphisms $A\to B$. Enrichment is the study of the
structure of those sets of morphisms, which could be vector spaces or
topological spaces for example, more details can be found in~\cite{KELLY196515,
kelly1982basic, maranda_1965}. It turns out that homsets in join inverse rig
categories are dcpos -- directed-complete partial orders, i.e. a partial-ordered
set with all directed joins. This allows us to consider fixpoints in
homsets. $\DCPO$ is the category of directed complete partial orders and
Scott-continuous functions -- monotone functions preserving joins. Dcpos are
often used for the denotational interpretation of different sorts of
$\lambda$-calculi, and more generally, to interpret recursive functions or
indefinite loops.

\begin{definition}[\cite{fiore04axiomatic}]\label{def:dcpocat}\rm
	A category enriched over $\DCPO$, also called a $\DCPO$-category, is a
	locally small category whose hom-sets are directly partial ordered and where
	composition is a continuous operation (\textit{i.e.} a morphism in $\DCPO$).
\end{definition}

It is proven in \cite{kaarsgaard2017join} that a join inverse category
can be considered enriched in $\DCPO$ without loss of generality.

\begin{lemma}\label{lem:dcpofun}
  Let $\CC$ be a join inverse rig category. The functors:
    $-\otimes - \colon \CC\times\CC\to\CC$,
    $-\oplus - \colon \CC\times\CC\to\CC$,
    $-\inv \colon \CC^{op}\to\CC$
  are $\DCPO$-functors, meaning that they preserve the dcpo structure
	of homsets.
\end{lemma}

\subsection{Compactness}

Inductive data types are written in the syntax as some least fixed point. As
said earlier, types are represented as objects in the category, and thus a type
judgement is an object mapping, or rather an endofunctor.  Here, we show how to
consider fixed points of endofunctors in our categorical setting.

\begin{definition}[Initial Algebra]\rm
	Given an endofunctor $F \colon \CC\to\CC$, an $F$-algebra is a pair
	of an object $A$ and a morphism $f \colon FA\to A$. $F$-algebras form
	a category with $F$-algebras homomorphisms. An initial $F$-algebra is
	an initial object in the category of $F$-algebras.
\end{definition}

\begin{theorem}[Lambek's theorem]
	Given an endofunctor $F:\CC\to\CC$ and an $F$-initial algebra
	$(X,\alpha\colon FX\to X)$, $\alpha$ is an isomorphism.\qed
\end{theorem}

With Lambek's theorem, we know that an initial algebra provides an object $X$
such that $X\cong FX$; $X$ is a fixed point of the endofunctor $F$, as
requested. The existence of such fixed points is given by the next theorem
\cite[Corollary~7.2.4]{fiore04axiomatic}.

\begin{definition}[Ep-pair]
  \label{def:ep-pair}
  Given a $\dcpo$-category $\CC$, a morphism $e \colon X \to Y$ in $\CC$ is
  called an \emph{embedding} if there exists a morphism $p \colon Y \to X$ such
  that $p \circ e = \iid_X$ and $e \circ p \leq \iid_Y$. The morphisms $e$ and
  $p$ form an \emph{embedding-projection pair} $(e,p)$, also called
  \emph{ep-pair}.
\end{definition}

We recall that an \emph{ep-zero} \cite[Definition~7.1.1]{fiore04axiomatic},
is an initial object such that every morphism with it as source is an embedding,
and is also a terminal object such that every morphism with it as target is a
projection.

\begin{theorem}
	\label{th:param-alg}
	A $\DCPO$-category with an ep-zero and colimits of $\omega$-chains
	of embeddings is parametrised $\DCPO$-algebraically $\omega$-compact;
	meaning that for every $\DCPO$-functor
	$F\colon\CC\times\DD\to\DD$, there is a pair consisting of
	a $\DCPO$-functor $F^\noma\colon \CC\to\DD$ and an indexed family
	$\alpha^F=\{\alpha^F_A \colon F(A,F^\noma A)\to F^\noma A \}$ of
	initial $F(A,-)$-algebras.
	This pair is called a \emph{parametrised} initial algebra.
        \qed
\end{theorem}

The hypotheses of the theorem above are verified by the categories we want to work with, without loss of generality.

\begin{proposition}[\cite{kaarsgaard2017join}]
	\label{prop:omega-chains}
	Any join inverse rig category can be faithfully embedded in a rig join
	inverse category with colimits of $\omega$-chains of embeddings. \qed
\end{proposition}

\section{Denotational semantics}
\label{sec:detonational}

We now show how to build a denotational semantics for the language we presented
thus far. The semantics is akin to the one presented in~\cite{nous21invcat} but
with extra structure to handle inductive types and recursive functions. While
the semantics is sound and adequate w.r.t. a notion of operational equivalence
between terms, the main interest of the semantics rest in showing that, given
some RTM $M$ whose semantics is a function $f$, we show that the semantics of
$\opn{isos}(M)$ is the same as $f$. This would provide us with a formal proof
that any computable reversible function can be captured by an iso.

\medskip
\noindent
\textsf{\bfseries Types.~}
Let us consider $\CC$ a join inverse rig category
(Definition~\ref{def:rig}).  We can assume without loss of generality
that $\CC$ satisfies the hypothesis of Theorem~\ref{th:param-alg}. In
order to deal with open types, we make use an auxiliary judgement for
types, of the form $X_1,\ldots,X_n\vDash A$, where $\{X_i\}_i$ is a
subset of the free type variables appearing in $A$. We interpret this
kind of judgement as a $\DCPO$-functor $\CC^{\abs\Theta}\to\CC$
written $\sem{\Theta\vDash A}$. This can be formally defined as a
(simple) inductive relation,
and the semantics is defined similarly to what is done in
\cite{fiore04axiomatic,zamdzhiev2021commutative}.
$\sem{\Theta\vDash\1}$ is the constant functor that maps to the tensor product
unit. $\sem{\Theta,X\vDash X}$ is a projection. The other judgements are
obtained by induction: if $\sem{\Theta\vDash A} = f$ and $\sem{\Theta\vDash B}
= g$, then $\sem{\Theta\vDash A\oplus B} = \oplus\circ\pv{f}{g}$ and
$\sem{\Theta\vDash A\otimes B} = \otimes\circ\pv{f}{g}$. Finally,
$\sem{\Theta\vDash \mu X.A} = \left(\sem{\Theta,X\vDash A}\right)^\noma$.
All this is summed up in Table~\ref{fig:rev-type-interpretation}.

\begin{table}[!h]
  \begin{align*}
    \sem{\Theta \vDash A} &\colon \CC^{\abs\Theta}\to\CC \\
		\sem{\Theta, X \vDash X} &= \Pi \\
    \sem{\Theta \vDash I} &= K_1 \\
		\sem{\Theta\vDash A\oplus B} &= \oplus\circ\pv{\sem{\Theta\vDash
		A}}{\sem{\Theta\vDash B}} \\
		\sem{\Theta\vDash A\otimes B} &= \otimes\circ\pv{\sem{\Theta\vDash
		A}}{\sem{\Theta\vDash B}} \\
		\sem{\Theta\vDash \mu X.A} &= \left(\sem{\Theta,X\vDash A}\right)^\noma
  \end{align*}
  \caption{Interpretation of types.}
  \label{fig:rev-type-interpretation}
\end{table}

Lemma~\ref{lem:dcpofun} and Theorem~\ref{th:param-alg} ensure that this is
well-defined. For closed types, we have
  $\den\one = 1$,
  $\den{A\oplus B} = \den A \oplus \den B$,
 $\den{A\otimes B} = \den A \otimes \den B$ and
 $\sem{\mu X.A} \cong \sem{A[\mu X.A/X]}$.
Ground iso types are represented by dcpos of morphisms in $\CC$, written
$\den{A\iso B} = \Hom_{\CC}(\den A, \den B)$. The type of iso functions $T_1
\to T_2$ is interpreted by the dcpo of Scott continuous maps between the two
dcpos $\sem{T_1}$ and $\sem{T_2}$, written $[\sem{T_1} \to \sem{T_2}]$.
The terms used to build isos are dependent in two contexts: variables in
$\Delta$ and isos in $\Psi$. In general, if $\Delta = x_1 \colon A_1, \dots,
x_m \colon A_m$ and $\Psi = \phi_1 \colon T_1, \dots, \phi_n \colon
T_n$, then we set
  $\den\Delta = \den{A_1}\otimes\dots\otimes\den{A_m}$ and
	$\den\Psi = \sem{T_1} \times\dots\times \sem{T_2}$,
with $\otimes$ being the monoidal product in $\CC$ and $\times$ the cartesian
product in $\DCPO$.

\begin{table}
	\begin{align*}
		\sem{\Psi ; \Delta \entail t \colon A}(g) &\in \CC(\sem\Delta, \sem A) \\
		\sem{\Psi ; \emptyset \entail * \colon I}(g)
		&= \iid_{\sem I} \\
		\sem{\Psi ; x \colon A \entail x \colon A}(g)
		&= \iid_{\sem A} \\
		\sem{\Psi ; \Delta \entail \inl t \colon A \oplus B}(g)
		&= \iota_l \circ \sem{\Psi ; \Delta \entail t \colon A}(g) \\
		\sem{\Psi ; \Delta \entail \inr t \colon A \oplus B}(g)
		&= \iota_r \circ \sem{\Psi ; \Delta \entail t \colon B}(g) \\
		\sem{\Psi ; \Delta_1, \Delta_2 \entail t_1 \otimes t_2 \colon A \otimes B}(g)
		&= \sem{\Psi ; \Delta_1 \entail t_1 \colon A}(g) \otimes
		\sem{\Psi ; \Delta_2 \entail t_2 \colon B}(g) \\
		\sem{\Psi ; \Delta \entail \fold t \colon \mu X . A}(g)
		&= \alpha^{\sem{X \vDash A}} \circ \sem{\Psi ; \Delta \entail t \colon A[\mu
		X . A / X]}(g)
	\end{align*}
	\begin{align*}
		\sem{\Psi \entailiso \omega \colon T}
		&\in \dcpo(\sem\Psi, \sem T) \\
		\den{\Psi,\phi \colon T \entailiso \phi \colon T} &=
		\pi_{\sem T} \\
		\sem{\Psi \entailiso \omega_2 \omega_1 \colon T_2}
		&= \rmeval \circ \pv{\sem{\Psi \entailiso \omega_2 \colon T_1 \to
		T_2}}{\sem{\Psi \entailiso \omega_1 \colon T_1}} \\
		\sem{\Psi \entailiso \lambda \phi . \omega \colon T_1 \to T_2} &=
		\rmcurry(\sem{\Psi, \phi \colon T_1 \entailiso \omega \colon T_2}) \\
		\sem{\Psi \entailiso \ffix \phi . \omega \colon T}
		&= \fix (\sem{\Psi, \phi \colon T \entailiso \omega \colon T})
	\end{align*}
\caption{Denotational semantics of the language in a join inverse rig $\DCPO$-category.}
\label{tab:sem}
\end{table}

\medskip
\noindent
\textsf{\bfseries Terms.~}  A well-formed term judgement
$\Psi ; \Delta \entail t \colon A$ has for semantics a Scott
continuous map
$\sem{\Psi ; \Delta \entail t \colon
  A}\in\dcpo(\sem\Psi,\CC(\sem\Delta,\sem A))$, defined as in
\Cref{tab:sem} when $g \in \sem\Psi$.  All this is well-defined in
$\DCPO$ provided that $\sem{\Psi \entailiso \omega \colon A \iso B}$
is. This last point is the focus of the next section.

\begin{lemma}
  \label{lem:rev-orthogonal-semantics}
  Given two judgements $\Psi ; \Delta_1 \entail t_1 \colon A$ and
  $\Psi ; \Delta_2 \entail t_2 \colon A$, such that $t_1~\bot~t_2$, we
  have for all $g \in \sem\Psi$ the equality
  $ \sem{t_1}(g)\pinv \circ \sem{t_2}(g) =
  0_{\sem{\Delta_2},\sem{\Delta_1}}.  $\qed
\end{lemma}

\noindent
\textsf{\bfseries Isos.~}  Isos do only depend on function variables,
but they are innately morphisms, so their denotation will be similar
to terms ---a Scott continuous map.  We define the denotation of an
iso by induction on the typing rules.  The interpretation of an
iso-variable is direct, it is the projection on the last component.
The interpretations of evaluations and $\lambda$-abstractions are
usual in a cartesian closed category, in our case, $\dcpo$. All the
rules apart for the iso-abstraction are found in \Cref{tab:sem}. The
remaining rule, building an iso abstraction $\isobreduit$, needs more
details.

\begin{lemma}
  \label{lem:compati-clauses}
  Given a well-formed iso abstraction
  $\Psi \entailiso \isobreduit \colon A \iso B$, for all
  $g \in \sem\Psi$, the morphisms in $\CC$ given by
  $
    \sem{\Psi ; \Delta_i \entail e_i \colon B}(g)
    \circ \sem{\Psi ; \Delta_i \entail v_i \colon A}(g)\pinv,
  $
  with $i \in I$ are pairwise inverse compatible.\qed
\end{lemma}

Each clause $v_i \iso e_i$ of an iso abstraction is given an interpretation
$\sem{e_i} \circ \sem{v_i}\pinv$. The previous lemma shows that in the case of
an iso abstraction, the interpretations of all clauses can be joined (in the
sense of Definition~\ref{def:join}). This join also generalises to the join in
$\dcpo$ as shown by the lemma below.

\begin{lemma}
  \label{lem:join-scott}
  Given a dcpo $\Xi$, two objects $X$ and $Y$ of $\CC$, a set of
  indices $I$ and a family of Scott continuous maps
  $\xi_i \colon \Xi \to \CC(X,Y)$ that are pairwise inverse
  compatible, the function $\bigvee_{i \in I} \xi_i:\Xi \to \CC(X,Y)$
  defined by $x \mapsto \bigvee_{i \in I} \xi_i(x)$ is Scott continuous.
  \qed
\end{lemma}

The interpretation of an iso abstraction is then given by:
\begin{equation*}
	\begin{split}
		\den{\Psi\entailiso \isobreduit \colon A\iso B} =
		\bigvee_{i \in I} (\comp \circ \pv{
		\sem{\Psi ; \Delta_i \entail e_i \colon B}
		}{
			\sem{\Psi ; \Delta_i \entail v_i \colon
			A}\pinv
		})
	\end{split}
\end{equation*}
The semantics is well-defined, in the sense that the interpretation of
$\Psi \entailiso \isobreduit \colon A \iso B$ is a Scott continuous
map between the dcpos $\sem\Psi$ and $\CC(\sem A, \sem B)$.

\section{Adequacy}
\label{sec:rev-adequacy}
\label{sub:rev-sound}
\label{sub:rev-adequacy}

We show a strong relationship between the operational semantics and the
denotational semantics of the language.
First, we fix a mathematical interpretation $\sem -$ in a
join inverse rig category $\CC$, that is $\dcpo$-enriched and whose
objects $0$ and $1$ are distinct.

Since the language handles non-termination, our adequacy statement
links the denotational semantics to the notion of termination in the
operational semantics: Given $~\vdash t\colon A$, $t$ is said to be
\emph{terminating} if there exists a value $v$ such that
$t\rightarrow^* v$. We either write $t\downarrow$, or $t \downarrow
v$.

\begin{theorem}[Adequacy]
  \label{th:rev-adeq}
  Given $~\vdash t\colon A$, $t\downarrow$ iff
  $\sem{\entail t \colon A} \neq 0_{\sem A}$.
\end{theorem}

\smallskip
\noindent
\textsf{\bfseries Soundness.~}
%
We start by showing the simple implication in
Theorem~\ref{th:rev-adeq} amount to soundness: the denotational
semantics is stable w.r.t. computation.

\begin{proposition}[Soundness]
  \label{prop:rev-soundness}
	Given a valid term judgement $~\entail t \colon A$,
  provided that $t \to t'$, then we have
	$
		\sem{\entail t \colon A} = \sem{\entail t' \colon A}.
	$\qed
\end{proposition}

We can conclude that if
$~\entail t \colon A$ with $t \downarrow$, we have
$\sem{\entail t \colon A} \neq 0_{\sem A}$.
This shows one of the implications in \Cref{th:rev-adeq}. For the proof
of the other implication, we follow a syntactic approach, inspired by
the proof in \cite{pagani2014quantitative}.

\smallskip
\noindent
\textsf{\bfseries Proof of Adequacy.~}
Our proof of adequacy involves a finitary sublanguage, where the
number of recursive calls is controlled syntactically: instead of
general fixpoints, we introduce a family of finitary fixpoints
$\nfix n \phi . \omega$, unfolding $n$ times before reducing to the
empty iso $\{\}$, corresponding to the diverging iso.

We show the adequacy result for the finitary terms thanks to strong
normalisation, and then show that it implies adequacy for the whole
language; this is achieved by observing that a normalising finitary
term is also normalising in its non-finitary form.

\section{Semantics preservation}
\label{sec:sem-preservation}

In this section, we fix the interpretation $\sem -$ of the language in $\PInj$,
the category of sets and partial injections. This choice comes without any loss
of generality (see \cite{kastl1979inverse}), and allows us to consider
\emph{computable} functions. In this section, we show that given a computable,
reversible function $f \colon \interp{A} \to \interp{B}$, there exists an
iso $\omega \colon A\iso B$ such that $\interp{\omega} = f$. In order to do
that, we fix a canonical flat representation of our types.

\subsection{A Canonical Representation}

We define a canonical representation of closed values of some type $A$ into a new
type $\opn{Enc} = \mathbb{B} \oplus \one \oplus \one \oplus \one \oplus \one \oplus \natT$
(recall that $\mathbb{B} = \one\oplus \one$ and $\natT = \mu X. \one\oplus X$).
For simplicity let us name each the following terms of type $\opn{Enc}$ :
$\ttt = \inl{(\inl{\unit})}$,
$\fff = \inl{(\inr{\unit})},$
$S = \inr{(\inl{\unit})},$
$D^\oplus = \inr{(\inr{(\inl{\unit})})},$
$D^\otimes = \inr{(\inr{(\inr{(\inl{\unit})})})},$
$D^\mu = \inr{(\inr{(\inr{(\inr{(\inl{\unit})})})})}$,
and for every natural number $n$, we write $\tilde{n}$ for the term
$\inr{(\inr{(\inr{(\inr{(\inr{(\inr{(\overline{n})})})})})})}$, where
$\overline{n}$ is the encoding of natural numbers, as given
in~\Cref{ex:cantor}.
Now, given some closed type $A$, we can define $\floor{-}_A \colon A\iso
[\opn{Enc}]$ the iso that transform any close value of type $A$ into a list of
$\opn{Enc}$. The iso is defined inductively over $A$: $\floor{-}_\one = \{\unit \iso [S]\},$ and
  \[\floor{-}_{A\oplus B} = \left\{\begin{array}{lcl}
    \inl{(x)} & \iso & \letv{y}{\floor{x}_A}{D^\oplus :: \fff :: y} \\
    \inr{(x)} & \iso & \letv{y}{\floor{x}_B}{D^\oplus :: \ttt :: y}
  \end{array}\right\},\]
  \[\floor{-}_{A \otimes B} = \left\{\begin{array}{lcl}
    \pv{x}{y} & \iso & \letv{x'}{\floor{x}_A}{}
                       \letv{y'}{\floor{y}_B}{} \\
              & & \letv{\pv{z}{n}}{++~\pv{x'}{y'}}{}
                  D^\otimes :: \tilde{n} :: z
  \end{array}\right\},\]
  \[\floor{-}_{\mu X. A} = \left\{\begin{array}{lcl} \fold{x} & \iso &
    \letv{y}{\floor{x}_{A[\mu X. A/X]}}{} D^\mu :: y
    \end{array}\right\},\] where the iso $++ \colon [A]\otimes [A]
    \iso [A]\otimes \natT$ which concatenate two lists is defined as:
\[\fix f. \left\{\begin{array}{lcl}
  \pv{[]}{x} & \iso & \pv{x}{0} \\
  \pv{h :: t}{x} & \iso & \letv{\pv{y}{n}}{f~\pv{t}{x}}{}
  \pv{h :: y}{S(n)}
\end{array}\right\}.\]

\subsection{Capturing every computable injection}

With this encoding, every iso $\omega \colon A\iso B$ can be turned
into another iso $\floor{\omega} \colon [\opn{Enc}] \iso [\opn{Enc}]$
by composing $\floor{-}_A$, followed by $\omega$, followed by
$\floor{-}_B^{-1}$. This is in particular the case for isos that are
the images of a Turing Machine. We are now ready to see how every
computable function $f$ from $\interp{A} \to \interp{B}$ can be turned
into an iso whose semantics is $f$. Given a computable function $f
\colon \interp{A} \to \interp{B}$, call $M_f$ the RTM computing $f$.
Since $f$ is in $\PInj$, its output uniquely determines its input.
Following~\cite{bennett1973logical}, given the output of the machine
$M_f$ there exists another Turing Machine $M_f'$ which takes this
output and recover the original input of $M_f$.
In our encoding of a RTM, the iso will have another additional garbage
which consist of a natural number, i.e. the number of steps of the RTM $M_f$.
Using $\garRem{\opn{isos}(M_f)}{\opn{isos}(M_f')}$ we can obtain a single iso,
from the encoding of $A$ to the encoding of $B$, without any garbage left.
This also ensures that $\sem{\garRem{\opn{isos}(M_f)}{\opn{isos}(M_f')}}(x)
= (\sem{\opn{isos}(M_f)}(x))_1$, for any input $x$.

\begin{theorem}[Computable function as Iso]
	\label{th:computable}
	Given a computable function $f\colon \interp{A} \to \interp{B}$, let $g\colon
	\interp{[\opn{Enc}] \otimes [\opn{Enc}]} \to \interp{[\opn{Enc}] \otimes
	[\opn{Enc}]}$ be defined as $g = \interp{\floor{-}_B} \circ f \circ
	\interp{\floor{-}_A^{-1}}$, and let $\omega\colon A \iso B$ be defined as $\{x~\iso~
	\letv{y}{\floor{x}_A}{}
    \letv{y'}{\garRem{\opn{isos}(M_g)}{\opn{isos}(M_g')}~y}{}
    \letv{z}{\floor{y'}_B^{-1}}{z}\}$.
  Then $\interp{\omega} = f$.\qed
\end{theorem}

\section{Conclusion}

In this paper, we built upon the language presented
in~\cite{nous21invcat,chardonnet2022curry,sabry2018symmetric} in order
to represent any partial injective function which can manipulate
inductive types.  We showed how one can encode any Reversible Turing
Machine, hence the (reversible) Turing Completeness, and we gave a
denotational semantics based on join inverse rig categories, together
with a soundness and adequacy theorem. Most notably, we showed that for
any computable function $f$ from $\PInj$, there exists an iso whose
semantics is $f$, thus our language fully characterises all of the
computable morphisms in $\PInj$.

\bibliography{ref}

\newpage
\appendix

\section{Proofs of Section~\ref{sec:language}}

The matching predicate, written $\sigma[v] = v'$ is defined
in~\Cref{tab:pattern_matching}, and one can see that $\sigma[v] = v'$
if and only if $\sigma(v) = v'$.

\begin{table}[t]
  \centering
\[\infer{\sigma[x] = e}{\sigma = \{ x \mapsto e\}}
\quad
\infer{\sigma[\inl{(e)}] = \inl{(e')}}{\sigma[e] = e'}
\quad
\infer{\sigma[\inr{(e)}] = \inr{(e')}}{\sigma[e] = e'}
\quad
\infer{\sigma[\fold{(e)}] = \fold{(e')}}{\sigma[e] = e'}
\]
\[
\infer{
\sigma[\pv{e_1}{e_2}] = \pv{e'_1}{e'_2}
}{
\sigma_1[e_1] = e'_1
&
\sigma_2[e_2] = e'_2
&
\text{supp}(\sigma_1) \cap \text{supp}(\sigma_2) = \emptyset
&
\sigma = \sigma_1\cup\sigma_2
} \quad \infer{\sigma[\unit] = \unit}{}
\]
  \caption{Pattern-matching.}
  \label{tab:pattern_matching}
\end{table}

We can first show that two different orthogonal values cannot match a same
value:

\begin{lemma}
  \label{lem:isos:orthogonality}
  Given $S$, a finite set of values of type $A$ such that for
  all $v_1, v_2\in S, v_1\bot v_2$, and a value $v$ of type $A$, if
  there exists $v_1, v_2\in S$ and $\sigma_1, \sigma_2$ such that
  $\sigma_1(v_1) = v$ and $\sigma_2(v_2) = v$ then $v_1 = v_2$.\qed
\end{lemma}
\begin{proof}[Proof of Lemma~\ref{lem:isos:orthogonality}]
  By induction on $\sigma_1[v_1] = v$
  \begin{itemize}
      \item Case $\sigma_1[x] = v$: There is no value $v_2$ such
      that $v_2 \bot x$, hence $S = \set{x}$ so $v_2 = x$.
      \item Case $\sigma_1[\unit] = \unit$: Similarly.
      \item Case $\sigma_1[\inl{v_1'}] = \inl{v'}$: By definition of
      the pattern-matching we have $\sigma_1[v_1'] = v'$. The only
      possible value for $v_2$ is then $\inl{v_2'}$. By IH we get
      $v_1' = v_2'$ hence $v_1 = v_2$.
      \item The case for $\inr{v_1}$ and $\fold{v_1}$ are similar.
      \item Case for $\sigma[\pv{v_1^1}{v_1^2}] = \pv{v^1}{v^2}$: By
      definition of the pattern-matching we have $\sigma_1[v_1^1] =
      v^1$ and $\sigma_2[v_1^2] = v^2$. By orthogonality, we have
      that $v_2 = \pv{v_2^1}{v_2^2}$ and therefore by IH $v_1^1 =
      v_2^1$ and $v_1^2 = v_2^2$ so $v_1 = v_2$. \qed
  \end{itemize}
\end{proof}

\begin{proof}[Proof of Lemma~\ref{lem:inv-type}]
  By induction on $\omega$.
  \begin{itemize}
    \item $\omega = \isovar \colon T$. Therefore $\Psi = \isovar \colon T \cup
    \Psi'$ for some $\Psi'$. Then $\Psi^{-1} = \isovar \colon T^{-1} \cup
    \Psi'^{-1}$ and $\omega^{-1} = \isovar \colon T^{-1}$.

    \item If $\omega = \omega_1~\omega_2$, in this case we get $\omega_1$ is of
    type $(A\iso B) \to T$ while $\omega_2$ is of type $A\iso B$. Since
    $(\omega_1~\omega_2)^{-1} = \omega_1^{-1}~\omega_2^{-1}$ we get that
    $\omega_1^{-1} \colon (B\iso A) \to T^{-1}$ and $\omega_2 \colon B\iso A$ which
    can be proved by induction hypothesis as the context $\Psi$ is duplicated on
    both premises.

    \item If $\omega = \fix \isovar. \omega' \colon (A_i \iso B_i)_{i\in I} \to
    (A\iso B)$, then we obtain

    $\infer{\Psi \vdash_\omega \fix \isovar. \omega' \colon (A_i
    \iso B_i)_{i\in I} \to (A\iso B)}{\Psi, \isovar \colon A\iso B \vdash_\omega \omega' :
    (A_i \iso B_i)_{i\in I}}$.

    By definition, we get

    $\infer{\Psi^{-1} \vdash_\omega \fix\isovar. \omega'^{-1} \colon (B_i \iso A_i)_{i\in I}
    \to (B\iso A)} {\Psi^{-1}, \isovar \colon B\iso A \vdash_\omega \omega'^{-1} \colon (B_i \iso
    A_i)_{i\in I} \to (B\iso A)}$ which can then be proved by induction
    hypothesis.

    \item If $\omega = \lambda \isovar. \omega'$, is similar.

    \item If $\omega = \isobasique$, is direct as the typing contexts stays
    the same for the values and expressions and $OD$ is trivially preserved and
    by induction hypothesis on all isos subterms of $\omega$
    \qed
  \end{itemize}
\end{proof}

\begin{proof}[Proof of~\Cref{lem:inv-rewriting}]
  By a direct induction on $\to$, there is only two possible rewriting rule for
  an iso: either it is $(\lambda \isolambdavar. \omega_1) \omega_2$ which
  rewrites to $\omega_1[\omega_2/\isolambdavar]$ and whose inverse
  $(\lambda\isolambdavar. \omega_1^{-1}) \omega_2^{-1}$ rewrites to
  $\omega_1^{-1} [\omega_2^{-1}/\isolambdavar]$. The case where the
  iso is of the shape $\fix \isovar. \omega$ is also direct. All the congruence
  case are handled by induction hypothesis.
\end{proof}

\begin{lemma}[Commutativity of substitution~\cite{chardonnet2022curry}]
  \label{isos:lem:commutativity-substitution}
  Let $\sigma_1, \sigma_2$ and $v$, such that $\sigma_1\cup\sigma_2$
  closes $v$ and $\mathtt{supp}(\sigma_1) \cap \mathtt{supp}(\sigma_2) =
  \emptyset$ then $\sigma_1(\sigma_2(v)) = \sigma_2(\sigma_1(v))$
  \end{lemma}

  \begin{proof}
  Direct by induction on $v$ as $\sigma_1$ and $\sigma_2$ have disjoint
  support: In the case where $v = x$ then either $\set{x\mapsto v'} \in
  \sigma_1$ or $\set{x\mapsto v'}\in\sigma_2$ and hence
  $\sigma_1(\sigma_2(x)) = v' = \sigma_2(\sigma_1(x))$. All the other
  case are by direct induction hypothesis as the substitutions enter the
  subterms.
\end{proof}

\begin{proof}[Proof of Theorem~\ref{thm:isos-iso}]
  Without lost of generality we only consider the case where $\omega$ is a set
  of clauses, then by induction hypothesis on the size of $\omega$:
    \begin{itemize}
        \item Case where $\omega = \{v_1\iso v_1\alt \dots\alt v_n\iso v_n'\}$
        then $\omega^{-1}(\omega~v_0)$, by non-overlapping and the fact that the
        whole term reduce to a value, there exists a $v_i$ such that
        $\sigma[v_i] = v_0$ and hence the term reduces to $\omega^{-1}
        \sigma(v_i')$. It is clear that $\sigma[v_i'] = \sigma(v_i)$ and hence
        the terms reduces to $\sigma(v_i)$, but by the first pattern-matching we
        know that $\sigma(v_i) = v_0$, which concludes the case.

        \item Case where $\omega = \isobasique$,

    for simplicity of writing we
    write a single clause:

    \[\begin{array}{c} \left(
        \begin{array}{l@{~}c@{~}l} v_1&{\iso}&\letv{p_1}{\omega_1~p_1'}{} \\
                && \cdots \\
                && \letv{p_n}{\isoterm_n~p'_n}{v'_1}
        \end{array}
        \right)^{-1} := \left(
        \begin{array}{lcl@{}l@{}l} v'_1&{\iso}&\letv{p'_n}{\isoterm_n^{-1}~p_n}{} \\
                && \cdots \\
                && \letv{p'_1}{\isoterm_1^{-1}~p_1}{v_1}
        \end{array}
        \right).
        \end{array}
    \]

    Take some closed value $\entaile v_0 \colon A$ such that $\sigma[v_1] = v_0$.

    By linearity, we can decompose $\sigma$ into $\sigma_1, \dots, \sigma_n,
    \sigma_{n+1}$ such that, after substitution we obtain
    \[ \letv{p_1}{\omega_1\,\sigma_1(p'_1)}{} \dots
    \letv{p_n}{\omega_n\,\sigma_n(p'_n)}{}
    \sigma_{n+1}(v_1')\]

    Given By our hypothesis that the whole term reduces to a value, we know that
    each {\lett} construction will reduce, and by the rewriting strategy we get:
    \[\begin{array}{c} \letv{p_1}{\omega_1\,\sigma_1(p'_1)}{}
        \\
        \dots \\
        \letv{p_n}{\omega_n\,\sigma_n(p'_n)}{} \\
        \sigma_{n+1}(v_1') \end{array}
        \to \begin{array}{c} \letv{p_1}{\overline{v_1}}{}
            \\
            \dots \\
            \letv{p_n}{\omega_n\,\sigma_n(p'_n)}{} \\
            \sigma_{n+1}(v_1') \end{array}
        \to \begin{array}{c} \letv{p_2}{\omega_2\,\gamma_1^2(\sigma_1(p'_1))}{}
            \\
            \dots \\
            \letv{p_n}{\omega_n\,\gamma_1^n(\sigma_n(p'_n))}{} \\
            \sigma_{n+1}(v_1') \end{array} \] The final term reduces to
        $\gamma_n^n(\dots(\gamma_1^n(\sigma_{n+1}(v_1')))\dots)$
    and creates a new substitution $\gamma_i$, the term will hence reduce
    to $\gamma_n(\dots(\gamma_1(\sigma_{n+1}(v_1')))\dots)$. Let $\delta =
    \cup_i \gamma_i \cup \sigma_{n+1}$

    We are left to evaluate

    \[\left(
          \begin{array}{lcl@{}l@{}l} v_1'&{\iso}&
            \letv{p'_n}{\isoterm_n^{-1}~p_n}{} \\
               && \cdots \\
               && \letv{p'_1}{\isoterm_1^{-1}~p_1}{v_1}
          \end{array}
          \right)~\delta(v_1')\]

    We get $\delta[v_1'] = \delta(v_1')$.

    We know that each $\gamma_i$ closes only $p_i$, we can
    therefore substitute the term as:
    \[\letv{p_n'}{\omega_n~\gamma_n(p_n)}{\dots}
    \letv{p_1'}{\omega_1\gamma_1(p_1)}{\sigma_{n+1}(v_1')}\]

    By induction hypothesis, Each {\lett} clause will re-create the
    substitution $\sigma_i$, we know this as the fact that the initial
    {\lett} construction, $\letv{p_i}{\omega_i~\sigma_i(p_i')}{\dots}$
    reduces to $\letv{p_i}{v_i}{\dots}$ While the new one,
    $\letv{p_i'}{\omega^{-1}\gamma_i(p_i)}{\dots}$, is, by definition
    of the substitution the same as $\letv{p_i'}{\omega^{-1}
    v_i}{\dots}$

    Then, since we know that $v_i$ is the result of
    $\omega~\sigma_i(p_i')$, we get by induction hypothesis $\sigma(p_i')$
    as the result of the evaluation.

    Therefore, after rewriting we obtain:
    $\sigma_n(\dots(\sigma_1(\sigma_{n+1}(v_1')))\dots)$.
    By~\Cref{isos:lem:commutativity-substitution} we get
    $\sigma_1(\dots(\sigma_n(\sigma_{n+1}(v_1')))\dots)$ which is
    equal to $v$.
    \end{itemize}
\end{proof}

\begin{lemma}
  \label{lem:iso-subst-type}
  If $\Psi, \psi \colon T_1 \entailiso \omega_2 \colon T_2$ and $\Psi
  \entailiso \omega_1 \colon T_1$ are well-formed, then $\Psi \entailiso
  \omega_2[\omega_1/\psi] \colon T_2$.
\end{lemma}
\begin{proof}
  Formally, the inductive definition of an iso judgement also depends on term
  judgements, in other words we also prove the following statement:
  If $\Psi, \psi \colon T_1 ; \Delta \entail t \colon A$ and $\Psi
  \entailiso \omega_1 \colon T_1$ are well-formed, then $\Psi ; \Delta \entail
  t[\omega_1/\psi] \colon A$.
  The proof is done by mutual induction on the term and iso judgements.
  \begin{itemize}
    \item $\Psi, \psi \colon T_1 ; \emptyset \entail \unit \colon \one$. Direct.
    \item $\Psi, \psi \colon T_1 ; x \colon A \entail x \colon A$. Direct.
    \item $\Psi, \psi \colon T_1 ; \Delta_1, \Delta_2 \entail t_1 \otimes t_2 \colon A \otimes B$.
      We observe that $(t_1 \otimes t_2)[\omega_1/\psi] = t_1[\omega_1/\psi]
      \otimes t_2[\omega_1/\psi]$ and the induction hypothesis concludes.
    \item $\Psi, \psi \colon T_1 ; \Delta \entail \ini t \colon A_1 \oplus A_2$.
      Similar to the previous point.
    \item $\Psi, \psi \colon T_1 ; \Delta \entail \fold t \colon \mu X . A$.
      Similar to the previous point.
    \item $\Psi, \psi \colon T_1 ; \Delta \entail \omega~t \colon B$.
      We observe that $(\omega~t)[\omega_1/\psi] =
      \omega[\omega_1/\psi]~t[\omega_1/\psi]$ and the induction hypothesis
      concludes.
    \item $\Psi, \psi \colon T_1 ; \Delta_1, \Delta_2 \entail \letv{p}{t_1}{t_2} \colon B$.
      With the induction hypothesis, similar to the previous point.
    \item $\Psi, \psi \colon T_1 \entailiso \isobreduit \colon A \iso B$.
      By induction hypothesis, given any iso $\omega$ present in $e_i$,
      $\omega$ is well-formed and $\omega[\omega_1/\psi]$ is also.
    \item $\Psi, \psi \colon T_1, \phi \colon T \entailiso \phi \colon T$. Direct.
    \item $\Psi, \psi \colon T_1 \entailiso \ffix \phi . \omega \colon T$. Note that $(\ffix
      \phi . \omega)[\omega_1/\psi] = \ffix \phi . (\omega[\omega_1/\psi])$,
      and by induction hypothesis $\omega[\omega_1/\psi]$ is well-formed.
    \item $\Psi, \psi \colon T_1 \entailiso \lambda \phi . \omega \colon T_2 \to T'_2$.
      Similar to the previous point.
    \item $\Psi, \psi \colon T_1 \entailiso \omega' \omega \colon T_2$.  Note that $(\omega'
      \omega)[\omega_1/\psi] = \omega' [\omega_1/\psi] \omega[\omega_1/\psi]$,
      and by induction hypothesis, $\omega'[\omega_1/\psi]$ and
      $\omega[\omega_1/\psi]$ are well-formed.
  \end{itemize}
\end{proof}

\begin{lemma}[Iso Subject Reduction]
  \label{lem:iso-subject-reduction}
  If $\Psi \entailiso \omega \colon T$ is well-formed and $\omega \to \omega'$,
  then $\Psi \entailiso \omega' \colon T$.
\end{lemma}
\begin{proof}
  The proof is done by induction on $\to$.
  \begin{itemize}
    \item $\ffix \phi . \omega \to \omega[\ffix \phi . \omega/\phi]$.
      The iso $\ffix \phi . \omega$ is well-formed, thus $\omega$ is also,
      and the previous lemma concludes.
    \item $(\lambda \phi . \omega_1) \omega_2 \to \omega_1[\omega_2/\phi]$.
      For the application to be well-formed, we need both $\lambda \phi .
      \omega_1$ and $\omega_2$ to be well-formed.  The former ensures that
      $\omega_1$ is well-formed, and the previous lemma concludes.
    \item $\omega_1 \omega_2 \to \omega'_1 \omega_2$. The induction
      hypothesis ensures that $\omega'_1$ is well-formed, and $\omega_2$
      is also because the application $\omega_1 \omega_2$ is.
  \end{itemize}
\end{proof}

\begin{lemma}[Iso Progress]
  \label{lem:iso-progress}
  If $~\entailiso \omega \colon T$ is well-formed, $\omega$ is either an
  iso value or there exists $\omega'$ such that $\omega \to \omega'$.
\end{lemma}
\begin{proof}
  The proof is done by induction on $~\entailiso \omega \colon T$.
  \begin{itemize}
    \item $~\entailiso \isobreduit \colon A \iso B$ is an iso value.
    \item $~\entailiso \ffix \phi . \omega \colon T$ reduces.
    \item $~\entailiso \lambda \phi . \omega \colon T_1 \to T_2$ is an iso
      value.
    \item $~\entailiso \omega_2 \omega_1 \colon T_2$. By induction hypothesis,
      either $\omega_2$ is a value or it reduces. If it reduces, $\omega_2
      \omega_1$ reduces. If it is a value, it cannot be an iso abstraction:
      being applied to another iso, it must have a type $T_1 \to T_2$. Thus, it
      is of the form $\lambda \phi . \omega'_2$, and $(\lambda \phi .
      \omega'_2) \omega_1$ reduces.
  \end{itemize}
\end{proof}

\begin{corollary}
  \label{cor:iso-progress}
  If $~\entailiso \omega \colon A \iso B$ is well-formed, either there is some
  $\Delta_i \entail v_i \colon A$ and $\Delta_i \entail e_i \colon B$ such that
  $\omega = \isobasique$, or there exists $\omega'$ such that $\omega \to
  \omega'$.
\end{corollary}

A valuation $\sigma$ is said to be well-formed with regard to contexts $\Psi$
and $\Delta$ if for all $(x_i \colon A_i) \in \Delta$, we have $\set{x_i
\mapsto t_i} \subseteq \sigma$ and $\Psi ; \emptyset \entail t_i \colon A_i$.
We write $\Psi ; \Delta \Vdash \sigma$.

\begin{lemma}
  \label{lem:rev-term-subst}
  If $\Psi ; \Delta \entail t \colon A$ and $\Psi ; \Delta \Vdash \sigma$
  are well-formed, then $\Psi ; \emptyset \entail \sigma(t)
  \colon A$ is well-formed.
\end{lemma}
\begin{proof}
  The proof is done by induction on $\Psi ; \Delta \entail t \colon A$.
  \begin{itemize}
    \item $\Psi ; \emptyset \entail \unit \colon \one$. Nothing to do.
    \item $\Psi ; x \colon A \entail x \colon A$.
      Since $x \colon A \Vdash \sigma$ is well-formed, there is $~\entail t
      \colon A$ such that $\set{x \mapsto t} \subseteq \sigma$ and $\sigma(x) =
      t$.
    \item $\Psi ; \Delta_1, \Delta_2 \entail t_1 \otimes t_2 \colon A \otimes B$.
      The induction hypothesis ensures that $\sigma(t_1)$ and $\sigma(t_2)$ are
      well-formed, and thus $\sigma(t_1) \otimes \sigma(t_2) = \sigma(t_1
      \otimes t_2)$ also is.
    \item $\Psi ; \Delta \entail \ini t \colon A_1 \oplus A_2$.
      The induction hypothesis ensures that $\sigma(t)$ is well-formed, and
      thus $\ini \sigma(t) = \sigma(\ini t)$ is.
    \item $\Psi ; \Delta \entail \fold t \colon \mu X . A$.
      The induction hypothesis ensures that $\sigma(t)$ is well-formed, and
      thus $\fold \sigma(t) = \sigma(\fold t)$ is.
    \item $\Psi ; \Delta \entail \omega~t \colon B$.
      The induction hypothesis ensures that $\sigma(t)$ is well-formed, and
      thus $\omega~\sigma(t) = \sigma(\omega~t)$ is.
    \item $\Psi ; \Delta_1, \Delta_2 \entail \letv{p}{t_1}{t_2} \colon B$.
      The induction hypothesis ensures that $\sigma(t_1)$ and $\sigma(t_2)$ are
      well-formed, and thus $\letv{p}{\sigma(t_1)}{\sigma(t_2)}
      = \sigma(\letv{p}{t_1}{t_2})$ also is.
  \end{itemize}
\end{proof}

\begin{lemma}[Subject Reduction]
  \label{lem:rev-subject-reduction}
  If $\Psi; \Delta \entail t \colon A$ is well-formed and $t \to t'$,
  then $\Psi; \Delta \entail t' \colon A$ is also well-formed.
\end{lemma}
\begin{proof}
  The proof is done by induction on $\to$. It revolves around three
  quick observations: Lemma~\ref{lem:iso-subject-reduction},
  Lemma~\ref{lem:rev-term-subst}, and that if $t_2$ and
  $C_\to[t_1]$ is well-formed, then $C_\to[t_2]$ also is.
\end{proof}


\section{Details on the encoding of RTM \texorpdfstring{, (\Cref{sec:RTM})}{}}

\subsection{Encoding}

We are now ready to encode any RTM as a well-typed iso of the
language. We first need to choose an encoding for the set of states
and the set of tape symbols:

\begin{definition}[Encoding of States and Tape Symbols]\rm
  Given a non-empty finite set of states $Q = \{q_1, \dots, q_n\}$ and finite
  set of tape symbols $\Sigma = \set{s_1, \dots, s_m}$ with $b\in
  \Sigma$ and $q_s, q_f\in Q$, we define the types $Q^T$ and
  $\Sigma^T$ as respectively $\bigoplus^{n} \one$ and
  $\bigoplus^{m} \one$.
\end{definition}

We then directly get a one-to-one correspondence between elements of
$Q$ and $\Sigma$ with closed values of $Q^T$ and $\Sigma^T$. We write
$q^T$ (resp. $a^T$) for the closed value corresponding to the element
$q\in Q$ (resp. $a\in\Sigma)$.

\begin{property}
  \label{isos-rtm:encoding-orthogonality}
  Given two different states (resp. two letter) $s_1, s_2$ we get
  $s_1^T \bot s_2^T$.\qed
\end{property}

We then define the type of configuration in the obvious manner: a
configuration $(q, (l, s, r))$ will be encoded as an element of $Q^T
\otimes ([\Sigma^T] \otimes \Sigma \otimes [\Sigma^T])$. We follow the
intuition of Zippers~\cite{huet1997zipper}: a single list (here, the
tape of the RTM) is represented by a pair of lists $(l_1, l_2)$ where
the list $l_1$ is in reverse order.

\begin{definition}[Encoding of Configurations]\rm We define the type
  of configurations as $C^T = (Q^T\otimes ([\Sigma^T] \otimes \Sigma
  \otimes [\Sigma^T]))$. Given a configuration $C = (q, ((\epsilon,
  a_1, \dots, a_n), a, (a_1', \dots, a_m', \epsilon)))$, it is encoded
  as $\opn{isos}(C) = (q^T, ([a_n^T, \dots, a_1^T], a^T, [a_1'^T,
  \dots, a_m'^T]))$. For example, the standard configuration $C =
  (q_s, (\epsilon, b, [a_1, \dots, a_n]))$ is represented as
  $\opn{isos}(C) = (q_s^T, ([], b^T, [a_1^T, \dots, a_n^T]))$.
\end{definition}

The problem with our current language is that we do not allow infinite
data structures. While the type of list is infinite, its instances are
finite. Implicitly we would like to have the equation $[] = [b] = [b,
b] = [b, b, b] = \dots$ but this is not possible in our current
system. Instead, in order to simulate this infinite amount of blank
symbol on both sides of the tape during the evaluation, we artificially
grow the size of the two tapes on both ends by blank symbols between each
transition step. In order to do that, we need the following:

\begin{lemma}[Existence of ``snoc'' iso]
  \label{lemma:snoc-iso}
  There exists an iso $\snoc \colon [A] \otimes A \iso [A] \otimes A$
  such that $\snoc~(\pv{[a_1, \dots, a_n]}{a}) \to^* \pv{[a_1,
  \dots, a_n, a]}{a}$. \qed
\end{lemma}
\begin{proof}
  In order to implement $\snoc$ properly we need the two
  following isos :

  \[\begin{array}{ll}
    \textit{len}:[A]\iso[A] \otimes \natT \\
    = \fix \isovar.\left\{
      \begin{array}{l@{~}c@{~}l}
        [~] & {\iso} & ([~], 0) \\
        h :: t & {\iso} & \letv{(t', n)}{\isovar~t}{} \\
        & & (h :: t', S(n)) \\
      \end{array}
    \right\}
    \end{array}
    \]
    \[
    \begin{array}{ll}
      \textit{snoc'} \colon [A] \otimes A \otimes \natT \iso [A] \otimes A \otimes \natT \\
      = \fix \isovar.\left\{
    \begin{array}{l@{~}c@{~}l}
      ([~], x, 0) & {\iso} & \letv{(x_1, x_2)}{\dup_A^\emptyset x}{}\\
      & & ([x_1], x_2, 0) \\
      (h::t, x, S(n)) & {\iso} & \letv{(t', x', n')}{\isovar (t, x, n)}{} \\
      (h :: t', x', S(n'))
    \end{array}
  \right\}
    \end{array}\]

  The fact that both isos are well-typed is trivial, the non-overlapping
  condition is always respected for both isos by the argument of type
  $\natT$.

  We have \begin{itemize}
  \item $\len[v_1, \dots, v_n] \to^* [v_1, \dots, v_n],
  \ov{n}$
  \item If $|l| = \ov{n}$ then $\snocc ([v_1, \dots, v_n], v, \ov{n})
  \to^* ([v_1, \dots, v_n, v], v, \ov{n})$.
  \end{itemize}
  Both can be showed by induction on $n$.

  \begin{itemize}
    \item Case $0$, for $\len$ we simply have that
    $\len~[] \to ([], 0)$ by pattern-matching on the first
    clause. For $\snocc$ this is due
    to~\Cref{lem:duplication-semantics}.

    \item Case $n+1$, we have the input being $[v_1, \dots, v_n,
    v_{n+1}]$, for $\len$ by induction hypothesis on the
    recursive call we know that $\len [v_2, \dots, v_n, v_{n+1}]
    \to^* ([v_2, \dots, v_n, v_{n+1}], \ov{n})$, and in the second
    clause we simply add $v_1$ back into the list and return
    $\ov{n+1}$.
    A similar argument is made for $\snocc$.
  \end{itemize}

  We are now ready to define $\snoc$:

  \[
    \begin{array}{ll}
      \textit{snoc} \colon [A] \otimes A\iso [A] \otimes A \\
      = \fix \isovar.\left\{
    \begin{array}{l@{~}c@{~}l}
      (x, y) & {\iso} & \letv{(x', n')}{\len~x}{} \\
      & & \letv{(x'', y', n')}{\snocc~(x', y, n)}{} \\
      & & \letv{z}{\len^{-1}~(x'', n')}{} \\
      & & (z, y')
    \end{array}
  \right\}
    \end{array}\]

    Now, the fact that $\snoc~(\pv{[a_1, \dots, a_n]}{a}) \to^*
    \pv{[a_1, \dots, a_n, a]}{a}$ follow directly from the above proof
    for $\len$ and $\snocc$ and due to~\Cref{thm:isos-iso}.
\end{proof}

\begin{definition}[Erasure of blank]\rm
  \label{def:erase-blank-iso}
  We define $\opn{erase}_b \colon A \otimes \Sigma^T \iso A$ which erase
    its second argument when its value is $b^T$ as $\{\pv{x}{b^T} \iso
    x\}$.
\end{definition}

In practice the $\opn{erase}_b$ iso will be applied in order to remove
the second argument of the result of $\snoc$.

\begin{definition}[Growth iso]\rm \label{def:growth-iso} We define
  $\growth \colon [\Sigma^T] \otimes [\Sigma^T] \iso [\Sigma^T]
  \otimes [\Sigma^T]$ as (where $l, r$ are variables):

  $\left\{\begin{array}{ccc}
  (l, r) & {\iso} &
      \letv{\pv{l'}{b_1}}{\snoc\pv{l}{b^T}}{} \\
  & & \letv{\pv{r'}{b_2}}{\snoc\pv{r}{b^T}}{} \\
  & & \letv{l''}{\opn{erase}_b \pv{l'}{b_1}}{} \\
  & & \letv{r''}{\opn{erase}_b \pv{r'}{b_2}}{} \\
  & & (l'', r'')
  \end{array}\right\}$
\end{definition}

\begin{lemma}[Growth iso semantics]
  \label{lemma:growth-semantics}

  Given two list $l = [a_1, \dots, a_n]$ and $r = [a_1', \dots, a_m']$
  (where $a_i, a_j'$ are some value of type $\Sigma$)

  We have that \[\growth(l, r) \to^* ([a_1^T, \dots,
  a_n^T, b^T], [a_1'^T, \dots, a_m'^T, b^T])\]
\end{lemma}
\begin{proof}
  Direct by Lemma~\ref{lemma:snoc-iso} and
  Definition~\ref{def:erase-blank-iso} on both $l^T$ and $r^T$.
\end{proof}

We can finally encode the transition relation $\delta$ where each
element of $\delta$ is encoded as a clause between values. Then the
encoding of the whole transition relation is simply the iso that have
for clauses the encoding of each element of $\delta$.

\begin{definition}[Encoding of $\delta$]\rm
  Given a RTM $M = (Q, \Sigma, \delta, b, q_s, q_f)$ a relation $(q, r, q') \in
  \delta$, is encoded as a clause between values $v_1 \iso v_2$ of type $C^T
  \iso C^T$ defined by case analysis on $r$, where $x_1, x_2, y$ are variables.

  First define the expression with a hole which apply the growth iso
  to the left and right tape as $\gctx[-]$ as $\letv{(l,
  r)}{\growth~(x_1, x_2)}{[-]}$ Then we can define the encoding
  of the transitions :
  \begin{itemize}
    \item Transition $(q, \rightarrow, q')$ is encoded as the clause \\
    $(q^T, (x_1, y_1, y_2 :: x_2)) \iso \gctx[(q'^T, (y_1 :: l, y_2, r))]$

    \item Transition $(q, \leftarrow, q')$ is encoded as the clause \\
    $(q^T, (y_1::x_1, y_2,  x_2)) \iso \gctx[(q'^T, (l, y_1, y_2 :: r))]$

    \item Transition $(q, \downarrow, q')$ is encoded as the clause \\
    $(q^T, (x_1, y, x_2)) \iso \gctx[(q'^T, (l, y, r))]$.

    \item Transition $(q, (s, s'), q')$ is encoded as the clause \\
    $(q^T, (x_1, s^T, x_2)) \iso \gctx[(q'^T, (l, s'^T, r))]$
  \end{itemize}
  Then, we define $\opn{isos}(M)$ as the iso whose clauses are the
  encoding of each rule of $\delta$.
\end{definition}

Along with~\Cref{isos-rtm:encoding-orthogonality} and the backward determinism
of the RTM, we get that clauses in the resulting iso are orthogonal to one
another, and that $\opn{isos}(M)$ is well-typed.

\begin{lemma}[$\opn{isos}(M)$ is well-typed]
  \label{lem:preom-typed}
  Given a RTM $M$, then $\vdash_\omega \opn{isos}(M) \colon C^T\iso C^T$.\qed
\end{lemma}
\begin{proof}
  While the typing of variables is obvious, one needs to check that the
  orthogonality is satisfied.

  Given two clauses, by forward determinism of the RTM we know that
  the state $q^T$ of the clauses are all different from one another
  when the iso is encoding a move transition; hence they are all
  orthogonal by~\Cref{isos-rtm:encoding-orthogonality}. The only other
  case occurs during a rewriting transition, in which case the state
  is the same, but the letter read has to be different (again by
  forward determinism), and
  again~\Cref{isos-rtm:encoding-orthogonality} apply. This is true for
  both the right-hand-side and the left-hand-sided clauses.
\end{proof}

We now show that the encoding of $\delta$, when the configuration is
extended by some empty character, simulates a step of the RTM.

\begin{lemma}
  \label{lem:sim-trans}
  Given a RTM $M$, if $M \vdash C \rightsquigarrow C'$ then
  $\opn{isos}(M)(\opn{isos}(C)) \to^* \tilde{C}$. Where $\tilde{C}$ is
  the result of applying the $\growth$ iso to the left and right
  tape of $\opn{isos}(C')$.
  \qed
\end{lemma}
\begin{proof}[Proof of Lemma~\ref{lem:sim-trans}]

  By analysis on $M\vdash C\rightsquigarrow C'$: Given the
  configuration \\ $(q, ((\epsilon, a_1, \dots, a_n), s, (a_1', \dots,
  a_m', \epsilon)))$, \\ we have $\opn{isos}(C) = (q^T,(l_i, s^T,
  r_i))$ where $l_i = [a_1^T, \dots, a_n^T]$ and $r_i = [a_m'^T,
  \dots, a_1'^T]$.

  We then consider the two possible kinds of rule:

  The first thing to note is that if $n$ or $m = 0$, then $l_i$ (resp.
  $r_i$) is empty, but we won't be required to handle those cases as the
  $\growth$-iso will ensure that we never have an empty tape
  during the simulation.
  \begin{itemize}
    \item Considering a rule $(q, (s, s'), q')$, we get that $C' =
    (q', (\epsilon, a_1, \dots, a_n), s', (a_1', \dots, a_m',
    \epsilon))$.
    By definition of the encoding this corresponds to the clause \\
    $(q^T, (x_1, s^T, x_2)) \iso \gctx[(q'^T, (l, s'^T, r))]$, where
    $x_1, x_2, l, r$ are variables. By pattern matching we have that
    the first element of each pair matches since they are equal closed
    values. Then, by pattern-matching we obtain the substitution
    $\sigma \set{x_1 \mapsto l_i, x_2 \mapsto r_i}$. By definition of
    the substitution we obtain, $\sigma(\gctx[(q'^T, (l, s'^T, r))])$,
    which, by Lemma~\ref{lemma:growth-semantics} reduces to $(q'^T,
    (l', s'^T, r'))$ where $l' = [a_n^T, \dots, a_1^T, b^T]$ and $r' =
    [a_1'^T, \dots, a_m'^T, b^T]$, this whole configuration is the
    result of applying the $\growth$-iso to the left and right
    tape of $\opn{isos}(C')$.

    \item Considering a rule $(q, d, q')$ for $d\in\set{\leftarrow,
    \downarrow, \rightarrow}$, suppose $d = \leftarrow$ (the other
    case being similar), we get that $C' = (q', (\epsilon, a_1, \dots,
    a_n, s), a_1', (a_2', \dots, a_m', \epsilon))$.
    By definition of the encoding this corresponds to the clause \\
    $(q^T, (y_1 :: x_1, y_2,x_2)) \iso \gctx(q'^T, (l, y_1, y_2 ::
    r))$, which, as previously will match with $\sigma = \set{x_1
    \mapsto l_i, x_2\mapsto r_i, y_1 \mapsto s^T}$ which by
    pattern-matching and Lemma~\ref{lemma:growth-semantics} reduces to
    the configuration $(q'^T, (l_i', s', r_i'))$ where $l_i = [s^T,
    a_n^T, \dots, a_1^T, b^T]$, $s' = a_1'^T$ and $r_i = [a_2'^T,
    \dots, a_m'^T, b^T]$ which is indeed the result of applying the
    $\growth$-iso to the left and right tape of $\opn{isos}(C')$.
    \qed
  \end{itemize}
\end{proof}

\begin{corollary}
  Given a RTM $M$ and an initial configuration $C$ such that $M \vdash
  C \rightsquigarrow^{n} C'$ where $C'$ is the final configuration, we
  get that $(\underbrace{\opn{isos}(M), \dots, \opn{isos}(M)}_{n\text{
  times }}) C^T \to^* \tilde{C'}$ where $\tilde{C'}$ is the encoding
  of the final configuration $C'$ with additional $blank$ symbols on
  each side of the tape. Furthermore, at each step of the evaluation
  we get that the left and right tape is never empty.
\end{corollary}
\begin{proof}
  This is a direct implication of~\Cref{lem:sim-trans}. The
  non-emptiness of the tapes follow from the fact that before each
  simulation we apply the $\growth$ iso on the configuration.
\end{proof}

\begin{remark}
Notice that given a RTM $M$, $\opn{isos}(M)^{-1} \not=
\opn{isos}(M^{-1})$ due to the fact that $\opn{isos}(M)^{-1}$ will
apply $\growth^{-1}$ and this is not what we want.
\end{remark}

So far, $\opn{isos}(M)$ simulates exactly one evaluation step of
the Turing Machine. In order to simulate a full run, we need to define
an iterator iso. It takes as argument an iso of type $A\iso A\otimes
(\one\oplus\one)$, which acts on data of type $A$ until the
additional boolean is set to $\fc$.

\begin{definition}[Iterator Iso]\rm
  Let $\omega_{\opn{aux}} \colon (A\iso A\otimes \natT) \to (A\otimes
    (\one\oplus\one)\iso A\otimes\natT)$ be
    \[\omega_{\opn{aux}} = \lambda\isolambdavar'. \left\{\begin{array}{l@{~}c@{~}l} (y, \tc) &\iso& \letv{(z,
    n)}{\isolambdavar'~y}{(z, S~n)} \\
    (y, \fc) & \iso & (y, 0) \end{array}\right\}\]
  Defined $\It \colon (A\iso A \otimes (\one\oplus\one)) \to (A\iso
  A\otimes \natT)$, counting the number of times its iso argument has been iterated:
  \[\It = \lambda\isolambdavar. \fix \isovar. \left\{\begin{array}{l@{~}c@{~}l}
     x& \iso & \letv{y}{\isolambdavar~x}{}
     \letv{z}{(\omega_{\opn{aux}}~\isovar)~y}{}
     z
  \end{array}\right\}\]
\end{definition}

As the iterator iso require an iso of type $A \iso A \otimes (\one
\oplus \one)$, we need to modify the encoding of the Turing Machine:

\begin{definition}[Second encoding]
  We modify the encoding of the RTM $\opn{isos}(M)$ into
  $\opn{isos}_\boolT(C)$ where the encoding of $\delta$, which we now
  call $\delta_\boolT$, is such that if the transition lead into the
  final state of the RTM, then the boolean is set to $\fc$, otherwise
  it is sent to $\tc$.
\end{definition}

The new encoding does not change the orthogonality and hence the new
iso is well-typed:

\begin{lemma}
  For any RTM $M$ we have $\vdash_\omega \opn{isos}_\boolT(M) \colon C^T
  \otimes (\one\oplus\one)\iso C^T \otimes (\one\oplus\one)$.\qed
\end{lemma}

We can now show that the applying the Iterator iso to some standard
configuration with $\opn{isos}_\boolT(C)$ is the same as applying
multiple time $\opn{isos}(C)$.

\begin{lemma}[Semantics of $\It(\omega)$]
    \label{lem:rtm:it}
    Given some RTM $M$ and initial configuration $C$ then if
    $(\It~\opn{isos}_\boolT(M))~C^T \rightarrow^* (v, \ov{n})$ then
    $\underbrace{\opn{isos}(M), \dots, \opn{isos}(M)}_{n+1 \text{
    times}}~C^T \to^* v$, where $\ov{n}$ is the encoding of the
    natural number $n$ into an element of type $\natT$, as given
    in~\Cref{ex:cantor}. \qed
\end{lemma}

\begin{proof}[Proof of Lemma~\ref{lem:rtm:it}]
  Call $\It'$ the Iterator iso where the variable $\isolambdavar$
  has been substituted by $\opn{isos}_\boolT(M)$.

  Then, proceed by induction on $n$.
  \begin{itemize}
      \item $0$, Our hypothesis state that:
      \begin{align*}
        & (\It~\opn{isos}_\boolT(M))~C^T \\
        &\to^* \letv{y}{\opn{isos}_\boolT(M)~C^T}{\letv{z}{\omega_{\opn{aux}}~y}{z}} \\
        &\to^* \letv{z}{(\omega_{\opn{aux}}~\It')~(v, \fc)}{z} \\
        &\to^* (v, 0)
      \end{align*}

      Therefore, we know that $v$ correspond to a final state of the
      RTM and thus: $\opn{isos}(M)~C^T \to v$.

      \item $n+1$:
      Our hypothesis state that:

      \begin{align*}
        &(\It~\opn{isos}_\boolT(M))~C^T\\
        &\to \letv{y}{\opn{isos}_\boolT(M)~C^T}{(\letv{z}{(\omega_{\opn{aux}}~\phi)~y}{z})} \\
        & \to
        \letv{z}{\omega_{\opn{aux}}~(v_1, \tc)}{z}  \\&\equiv \letv{z}{\left\{\begin{array}{l@{~}c@{~}l}
            (y, \tc) &\iso& \letv{(z, n)}{\It'~y}{(z, S~n)} \\
            (y, \fc) & \iso & (y, S~0) \end{array}\right\}~(v_1,
            \tc)}{z} \\
        &\to \letv{y}{(\letv{(z, n)}{\It~v_1}{(z, S~n)})}{y}
      \end{align*}

      Then we know that $\It'~v_1$ will reduce to some $(v', n)$
      and so by IH we get that
      $\underbrace{\opn{isos}(M)~\dots~\opn{isos}(M)}_{n+1 \text{
      times}}~v_1 \to^* (v', \fc)$, so we get
      $\underbrace{\opn{isos}(M)~\dots~\opn{isos}(M)}_{n+2 \text{
      times}}~C^T \to^* v'$. \qed
  \end{itemize}
\end{proof}

The new encoding does not change the orthogonality and hence the new
iso is well-typed:

Finally, we get from from~\Cref{lem:sim-trans} and~\Cref{lem:rtm:it}
that if, from an initial configuration $C$, a RTM goes into the final
configuration $C'$ in $n+1$ steps, then the Iterator iso on
$\opn{isos}(M)$ with input the encoding of $C$ reduces to the pair of
the encodings of $C'$ and $n$.

\begin{theorem}[Simulation of String Semantics]
  \label{thm:rtm-iso-simulation}
  Let $M$ be a RTM. If  $M \vdash (q_s, (\epsilon, b, s))
	\rightsquigarrow^{n+1} (q_f, (\epsilon, b, s'))$, then \\
	$It(\opn{isos}_\boolT(M))~(q_s^T, ([b^T], b^T, s^T)) \rightarrow^*
	((q_f^T, ([b^T, \dots, b^T], b, [a_1^T, \dots, a_n^T, b^T, \dots,
	b^T])), \overline{n})$ where $s' = (a_1, \dots, a_n)$.
\end{theorem}
\begin{proof}
  The proof follow directly from~\Cref{lem:sim-trans} and~\Cref{lem:rtm:it}.
\end{proof}

If fact the converse is also true, i.e. if given some RTM $M$ and a
configuration $C$ such that $\opn{isos}(M)(\opn{isos}(C)) \to v$ where
$v$ is another configuration, then $M \vdash C \rightsquigarrow v$.

\begin{theorem}[Equivalence of RTM and isos]
  Given a RTM $M$ and some configuration $C$, if
  $\opn{isos}(M)~\opn{isos}(C) \to v$ then $M \vdash C
  \rightsquigarrow C'$ where $v = \opn{isos}(C')$ up to some blank
  symbols on the right and left side of the tape.
\end{theorem}
\begin{proof}
  The proof is direct as there is a one-to-one correspondence between
  clauses of $\opn{isos}(M)$ and element of $\delta$, then proceeds by
  analysis of the reduction of $\opn{isos}(M)~\opn{isos}(C) \to v$ and
  by analysis of $C$. Since $\opn{isos}(M)$ will add some blank
  symbols on each side of the tape of $C$, even if the tape is
  initially empty, the transition is still defined.
\end{proof}

Notice that at the end of the simulation, we do not exactly obtain the
tape, but the tape up-to some blank symbols on the left and on the
right. By~\Cref{lem:sim-trans} we know that the left-tape is just full
of blanks symbols while the right-tape contain the result.

Is it then possible to apply a $\opn{clean-up}$ iso which takes a
configuration with additional blank and remove them, leaving only the
end result.

\begin{definition}[Remove Blank Iso]

  Given an alphabet $\Sigma = \set{b, a_1, \dots, a_n}$ We define the
  remove-blank iso from $[\Sigma] \to [\Sigma] \otimes \natS$ as:
  \[  \rmBlank = \fix \isovar. \left\{\begin{array}{ccc}
    [] & {\iso} & ([], 0) \\
    b^T :: t & {\iso} & \letv{(t', n)}{\isovar~t}{(t', S(n))} \\
    a_1^T :: t & {\iso} & ((a_1^T :: t), 0) \\
    \vdots & \vdots & \vdots \\
    a_n^T :: t & {\iso} & ((a_n^T :: t), 0) \end{array}\right\}\]
\end{definition}
\begin{lemma}[Semantics of $\rmBlank$]
  \label{lem:semantics-removeBlank}
  Given $l = [b_1^T, \dots, b_m^T, a_1^T, \dots, a_n^T]$, st $a_1,
  \dots, a_n\in\Sigma$ and $a_1\not= b$ then $\rmBlank l
  \to^* ([a_1^T, \dots, a_n^T], m)$.
\end{lemma}
\begin{proof}
  This is done by induction on $m, n$.
  \begin{itemize}
    \item If $m = 0, n > 0$ we enter one of the clause $a_i^T :: t$ for some
    $a_i\not= b\in\Sigma$, in which case the list is just returned as it
    is, with $0$ as second argument.
    \item If $m = 0, n = 0$, then the list is empty and we enter the
    first clause.
    \item If $m > 0, n = 0$, then we can apply our induction
    hypothesis on the recursive call.
    \item Case $m, n > 0$ is similar to the previous case.
  \end{itemize}
\end{proof}

As the right tape contains the blank at the end, we need to first
reverse it in order to apply the $\rmBlank$ iso:

\begin{definition}[List reversal]

First define $\opn{rev}_{\text{aux}} \colon [A] \otimes [A] \iso [A]
\otimes [A]$ as

\[ \fix\isovar. \left\{\begin{array}{ccc}
  ([], y) & {\iso} & ([], y) \\
  (h :: t, y) & {\iso} & \letv{(h_1, h_2)}{\dup_A^\emptyset~h}{} \\
  & & \letv{(t_1, t_2)}{\isovar (t, h_2 :: y)}{} \\
  & & (h_1 :: t_1, t_2)
  \end{array}\right\} \]

  Then the inverse of a list is simply $\opn{rev} \colon [A] \iso [A]
  \otimes [A]$ $\opn{rev} = \set{x \iso \letv{(t_1,
  t_2)}{\opn{rev}_{\text{aux}}~(x, [])}{(t_1, t_2)}}$
\end{definition}
\begin{lemma}[Semantics of $\opn{rev}$]
  \label{lem:semantics-revlist}
  Given a list $l = [v_1, \dots, v_n]$ then $\opn{rev} = (l, [v_n,
  \dots, v_1])$.
\end{lemma}
\begin{proof}
  By simple induction on $l$, if it is empty then it is direct by
  definition of $\opn{rev}_{\text{aux}}$, if it is $v_1 :: t$ then by
  substitution and~\Cref{lem:duplication-semantics} we have that the
  value that will be substituted to $h_1$ and $h_2$ is $v_1$, then we
  can apply our induction hypothesis.
\end{proof}

Then we can define the $\cleanUp$ iso from $C^T \otimes \natT \iso C^T \otimes
\natT \otimes \natT \otimes \natT \otimes [\Sigma^T]$
where, in the output type we have that
\begin{itemize}
  \item The configuration is the same as the input, without the extra
  blanks symbol on the tapes.
  \item The first natural number is the number of steps of the token
  machine, i.e. the same as the natural number given as input.
  \item The second natural number is the number of blank symbols on the left-tape.
  \item The third natural number is the number of blank symbols on the right-tape.
  \item The list of letter of the alphabet is the initial right-tape.
\end{itemize}

Then we can define the clean up as:

\begin{definition}[CleanUp Iso]

\[  \left\{\begin{array}{ccc}
  ((x, (l, y, r)), n) & {\iso} &
      \letv{(l', n_1)}{\rmBlank~l}{} \\
  & & \letv{(r_{\text{ori}}, r_{\text{rev}})}{\opn{rev}~r}{} \\
  & & \letv{(r', n_2)}{\rmBlank~r_{\text{rev}}}{} \\
  & & (x, (l', y, r')), n, n_1, n_2, r_{\text{ori}}
  \end{array}\right\}\]
\end{definition}

\begin{theorem}[CleanUp Semantics]
  \label{thm:cleanup}
Given a RTM $M$, and initial configuration $C$ such that $M \vdash C
\rightsquigarrow C' = (q, (\epsilon, b, (a_1, \dots, a_n)))$ and a value
$v$ of type $C^T$ such that\\ $v = (q^T, ([b^T, \dots, b^T], b^T,
[a_1^T, \dots, a_n^T, b^T, \dots, b^T]))$

Then we have that $\cleanUp \It(\opn{isos}_\boolT(M)) C^T
\to^* ((q^T, ([], b^T, [a_1^T, \dots, a_n^T])), z)$. where $z$ is of
type $\natT \otimes \natT \otimes \natT \otimes [\Sigma^T]$
\end{theorem}
\begin{proof}
  The proof is direct following~\Cref{thm:rtm-iso-simulation},
  \Cref{lem:semantics-removeBlank} and \Cref{lem:semantics-revlist}.
\end{proof}

The configuration obtained correspond exactly to $\opn{isos}(C')$ but
we now have to deal with another form of garbage. Thankfully, there
exists a systemic way to remove the garbage from a reversible process.


\subsection{Getting rid of the garbage}
As we mentioned, often in order to realise some operation in a
reversible way, one need to add some additional garbage. Nevertheless,
there exists a way to remove this garbage, albeit at the cost of a
heavier computation time. First described
in~\cite{bennett1973logical}, given two Turing machines $f_1$ and
$f_2$ and some input $\opn{in}$ such that if
$f_1(\opn{in}) = \opn{out} \otimes \opn{garbage}$ and
$f_2(\opn{out}) = \opn{in} \otimes \opn{garbage'}$, then the process
consist of taking additional tapes in the Turing Machine in order to
reversibly duplicate (represented by the $\oplus$) or reversibly erase
some data (represented by the $\chi$) in order to recover only the
output of $f_1$, without any garbage. The process is shown
in~\Cref{fig:garbage-removal}.

Given an iso $\omega \colon A\iso B\otimes C$ and
$\omega' \colon B \iso A \otimes C'$ where $C, C'$ represent garbage, we
can build an iso from $A\iso B$ as follow where the variables
$x, y, z$ (and their indices) respectively correspond to the first,
second, and third wire of~\Cref{fig:garbage-removal}. This operator
makes use of the iso \opn{Dup} discussed in Section~\ref{sec:dup}.

\begin{definition}[Garbage Removal Iso]~\rm
  We defined $\garRem{\omega}{\omega'}$ by $\{x_1~\iso~
\letv{\pv{x_2}{y}}{\omega~x_1}{}
             \letv{\pv{x_3}{z}}{\dup_B^\emptyset~x_2}{}
            \letv{x_4}{\omega^{-1}~\pv{x_3}{y}}{}
            \letv{\pv{z_2}{y_2}}{\omega'~z}{}
             \letv{z_3}{(\dup_B^\emptyset)^{-1}~\pv{z_2}{x_4}}{}
             \letv{z_4}{\omega'^{-1}~\pv{z_3}{y_2}}{}
             z_4\}$.
\end{definition}

\begin{lemma}[Garbage Removal is well-typed]
  \label{lem:garbage-removal-typed}
  Given $\omega \colon A\iso B\otimes C$ and $\omega' \colon B \iso A \otimes C'$ then
  $\garRem{\omega}{\omega'}$ is a well-typed iso of type $A\iso B$.\qed
\end{lemma}

\begin{lemma}[Garbage Removal Semantics]\label{lem:garb-remove-sem}
  Given $\omega \colon A\iso B\otimes C$ and $\omega' \colon B \iso A \otimes C'$ and any
  well-typed value $v$ of type $A$, such that if $\omega~v \to^* \pv{v_1}{v_2}$
  and $\omega'~v_1\to^* \pv{v_3}{v_4}$ then $v_3 = v$,
  then if $\garRem{\omega}{\omega'}~v \to^* v_5$ then $v_1 = v_5$.
	\qed
\end{lemma}

\subsection{Back to Reversible Turing Machines}

\begin{theorem}[Capturing the exact semantics of a RTM]
    For all RTM $M$ such that given standard configuration $C = (q_s,
    (\epsilon, b, s))$ and $C' = (q_f, (\epsilon, b, s'))$ such that
    $M \vdash C \rightsquigarrow^* C'$, we have that 
    \[\garRem{\opn{cleanUp}(\opn{It}
    (\opn{isos}_\boolT(M)))}{\opn{cleanUp}(\opn{It}
    (\opn{isos}_\boolT(M^{-1})))}~\opn{isos}(C) \to^* \opn{isos}(C')\]
\end{theorem}
\begin{proof}

    Write $\Omega_1 = \opn{cleanUp}(\opn{It} (\opn{isos}_\boolT(M)))$ and
    $\Omega_2 = \opn{cleanUp}(\opn{It}( \opn{isos}_\boolT(M^{-1})))$.

    Since we know that $M^{-1}$ computes the inverse of
    $M$~\cite{axelsen11rtm} and by~\Cref{lem:sim-trans},
    and~\Cref{thm:cleanup} we satisfy the pre-condition
    of~\Cref{lem:garb-remove-sem}.
\end{proof}

%

\section{Proofs of Section~\ref{sec:RTM}}

\begin{proof}[Proof of Lemma~\ref{lem:sim-trans}]
  By analysis on $M\vdash C\rightsquigarrow C'$: Given the
  configuration $(q, (l, s, r))$ where $l, r$ are the tapes and $s$
  the letter pointed by the head of the RTM. We have $\opn{isos}(C) =
  (q^T,(l^T, s^T::r^T))$.

  We then consider the two possible kind of rule:
  \begin{itemize}
    \item Considering a rule $(q, (s, s'), q')$, by definition this
    correspond to the clause $(q^T), (x_1, s^T::x_2) \iso (q'^T, (x_1,
    s'^T::x_2))$, where $x_1, x_2$ are variables. By pattern matching
    we have that the first element of each pair matches since they are
    equal, and closed values. Then, by pattern-matching, $l^T$ and
    $r^T$ matches respectively $x_1$ and $x_2$ and $s^T$ matches with
    $s^T$ as they are equal closed values. Thus, the term reduces to
    $\sigma((q'^T, (x_1, s'^T::x_2)))$ with $\sigma = \set{x_1 \mapsto
    l^T, x_2 \mapsto r^T}$ which is equal to $(q'^T, (l^T, s'^T::s^T))
    = \opn{isos}(C')$

    \item Considering a rule $(q, d, q')$ for $d\in\set{\leftarrow,
    \downarrow, \rightarrow}$, suppose $d = \leftarrow$ (the other
    case being similar), we get the corresponding clause $(q^T, (x_1,
    y::x_2)) \iso (q'^T, (y::x_1, x_2))$, which, as previously will
    match with $\sigma = \set{x_1 \mapsto l^T, x_2\mapsto r^T,
    y\mapsto s^T}$ and reduce to $\sigma((q'^T, (y::x_1, x_2))) =
    (q^T, (s^T::l^T, r^T)) = \opn{isos}(C')$.
  \end{itemize}
\end{proof}

\begin{proof}[Proof of Lemma~\ref{lem:rtm:it}]
  By induction on $n$.
  \begin{itemize}
      \item $0$, Our hypothesis state that:
      \begin{align*}
        & (\It~\omega)~v \\
        &\to \letv{y}{\omega~v}{\letv{z}{\omega_{\opn{aux}}~y}{z}} \\
        &\to \letv{z}{\omega_{\opn{aux}}~(v', \fc)}{z} \\
        &\to (v', 0)
      \end{align*}
      Therefore: $\omega~v \to (v', \fc)$.

      \item $n+1$:
      Our hypothesis state that:

      \begin{align*}
        &(\It~\omega)~v\\
        &\to \letv{y}{\omega~v}{(\letv{z}{\omega_{\opn{aux}}~y}{z})} \\
        & \to
        \letv{z}{\omega_{\opn{aux}}~(v_1, \tc)}{z}  \\&\equiv \letv{z}{\left\{\begin{array}{l@{~}c@{~}l}
            (y, \tc) &\iso& \letv{(z, n)}{\It(\omega)~y}{(z, S~n)} \\
            (y, \fc) & \iso & (y, S~0) \end{array}\right\}~(v_1,
            \tc)}{z} \\
        &\to \letv{y}{(\letv{(z, n)}{\It~v_1}{(z, S~n)})}{y}
      \end{align*}

      Then we know that $\It~v_1$ will reduce to some $v', n$ and so by IH
      we get that $\underbrace{\omega~\dots~\omega}_{n+1 \text{ times}}~v_1
      \to^* (v', \fc)$, so we get $\underbrace{\omega~\dots~\omega}_{n+2 \text{
      times}}~v \to^* (v', \fc)$.
  \end{itemize}
\end{proof}

\section{Proofs related to the denotational semantics}

\begin{table}[!h]
	\[
    \begin{array}{c}
			\infer{\Theta\vDash \1}{}
			\qquad
			\infer{\Theta, X\vDash X}{}
			\qquad
			\infer{\Theta\vDash \mu X.A}{\Theta,X\vDash A}
      \qquad
      \infer[\star\in\{\otimes,\oplus\}]{\Theta\vDash A\star B}
			{
				\Theta\vDash A
				&
				\Theta\vDash B
			}
		\end{array}
	\]
	\caption{Formation rules for types.}
	\label{tab:typesform}
\end{table}

\begin{lemma}
  \label{lem:res-zero-zero}
  Given an inverse category $\CC$ and a morphism $f \colon X \to Y$ such that
  $\res f = 0_{X,X}$, then $f = 0_{X,Y}$.
\end{lemma}
\begin{proof}
  By Def.~\ref{def:restr}, we know that $f = f \res f$, and thus $f = f \res f
  = f 0_{X,X} = 0_{X,Y}$.
\end{proof}

\begin{lemma}
  \label{lem:inv-ortho-compati}
  Given an inverse category $\CC$ and two morphisms $f, g \colon X \to Y$ in
  $\CC$ such that $f\pinv g = 0_{X,X}$ and $fg\pinv = 0_{Y,Y}$, then $f$ and
  $g$ are inverse compatible.
\end{lemma}
\begin{proof}
  Our goal is to prove that $f \res g = g \res f$, but we are going to prove
  that both are equal to zero. Def.~\ref{def:restr} ensures that $\res{f \res
  g} = \res f \res g$, and then Def.~\ref{def:invcat} gives that $\res f \res g
  = f\pinv f g\pinv g$, which is then equal to $f\pinv 0_{Y,Y} g$ by
  hypothesis, and thus $\res{f \res g} = \res f \res g = 0_{X,X}$;
  Lemma~\ref{lem:res-zero-zero} ensures then that $f \res g = 0_{X,Y}$. Note
  that $\res{g \res f} = \res f \res g = \res g \res f$ by
  Def.~\ref{def:restr}, thus $g \res f = 0_{X,Y}$ too. We have proven that $f
  \smile g$. The proof that $f\pinv \smile g\pinv$ is similar.
\end{proof}

\begin{lemma}
  \label{lem:orthogonal-injections}
  In a join inverse category $\CC$ with disjointness tensor, for all
  objects $X$ and $Y$, we have $\iota_l\pinv \circ \iota_r = 0_{Y,X}$.
\end{lemma}
\begin{proof}
  We know from Definition~\ref{def:disten} that $\rc{\iota_l}~\rc{\iota_r}
  = 0_{X \oplus Y}$, thus by postcomposition
  $\iota_l\pinv~\rc{\iota_l}~\rc{\iota_r} = 0_{X \oplus Y, X}$. Since $\CC$ is
  a restriction category (see Def.~\ref{def:restr}), $\iota_l\pinv~\rc{\iota_l}
  = \iota_l\pinv$, and thus $\iota_l\pinv~\rc{\iota_r} = 0_{X \oplus Y, X}$.
  With Definition~\ref{def:invcat}, we have $\rc{\iota_r} =
  \iota_r\iota_r\pinv$, thus $\iota_l\pinv \iota_r \iota_r\pinv = 0_{X \oplus
  Y, X}$. Finally, we precompose by $\iota_r$, to obtain $\iota_l\pinv \iota_r
  \iota_r\pinv \iota_r = 0_{Y,X}$. Definition~\ref{def:invcat} again tells us
  that $\iota_r\pinv ~\iota_r = \res{\iota_r}$, and also $\iota_r~\res{\iota_r}
  = \iota_r$, hence the equality: $\iota_l\pinv \iota_r = 0_{Y,X}$.
\end{proof}

\begin{proposition}
  Given a well-formed iso abstraction
  $\Psi \entailiso \isobreduit \colon A \iso B$, its interpretation
  $\sem{\Psi \entailiso \isobreduit \colon A \iso B}$ is well-defined
  as a Scott continuous map between the dcpos $\sem\Psi$ and
  $\CC(\sem A, \sem B)$.
\end{proposition}
\begin{proof}
This is a direct conclusion of Lemmas~\ref{lem:compati-clauses} and
 \ref{lem:join-scott}.
\end{proof}

\begin{lemma}
  \label{lem:sem-subst-iso}
  Given two well-formed isos $\Psi, \phi \colon T_2 \entailiso \omega_1 \colon
  T_1$ and $\Psi \entailiso \omega_2 \colon T_2$,
  the iso
   $\sem{\Psi \entailiso \omega_1[\omega_2/\phi] \colon T_1}$
   is defined as $\sem{\Psi, \phi \colon T_2 \entailiso \omega_1 \colon T_1}
    \circ \pv{\iid}{\sem{\Psi \entailiso \omega_2 \colon T_2}}$.
\end{lemma}
\begin{proof}
  The proof is done by induction on the formation rules for $\omega_1$.
\end{proof}

\begin{lemma}
  \label{lem:rev-orthogonal-semantics-values}
  Given two judgements $\Psi ; \Delta_1 \entail v_1 \colon A$ and
  $\Psi ; \Delta_2 \entail v_2 \colon A$, such that $v_1~\bot~v_2$, we have
  for all $g \in \sem\Psi$ the equality
  $
    \sem{v_1}(g)\pinv \circ \sem{v_2}(g) = 0_{\sem{\Delta_2},\sem{\Delta_1}}.
  $
\end{lemma}
\begin{proof}[Proof of \Cref{lem:rev-orthogonal-semantics-values}]
  This is proven by induction on the definition of $\bot$. The two cases $\inl
  v_1~\bot~\inr v_2$ and $\inr v_1~\bot~\inl v_2$ are covered by
  Lemma~\ref{lem:orthogonal-injections}. the other cases involve
  precompositions and tensor products, the result is direct with the induction
  hypothesis.
\end{proof}
\begin{proof}[Proof of \Cref{lem:rev-orthogonal-semantics}]
  A large part of this lemma is already proven in
  Lemma~\ref{lem:rev-orthogonal-semantics-values}. It remains to observe that
  the interpretation of $\mathtt{let}$ involves a precomposition, thus the
  induction hypothesis on the definition $\bot$ concludes.
\end{proof}

\begin{proof}[Proof of \Cref{lem:compati-clauses}]
  Lemma~\ref{lem:rev-orthogonal-semantics} gives us that for all $g \in
  \sem\Psi$ and $i \neq j \in I$:
  \begin{align*}
    (\sem{e_i}(g)
    \circ \sem{v_i}(g)\pinv)\pinv
    \circ
    (\sem{e_j}(g)
    \circ \sem{v_j}(g)\pinv)
    &= 0_{\sem A,\sem A} \\
    (\sem{e_i}(g)
    \circ \sem{v_i}(g)\pinv)
    \circ
    (\sem{e_j}(g)
    \circ \sem{v_j}(g)\pinv)\pinv
    &= 0_{\sem B,\sem B}
  \end{align*}
  which are the hypothesis of Lemma~\ref{lem:inv-ortho-compati} to
  ensure that for all $g \in \sem\Psi$ and $i \neq j \in I$,
  \[
    \sem{e_i}(g)
    \circ \sem{v_i}(g)\pinv
    \asymp
    \sem{e_j}(g)
    \circ \sem{v_j}(g)\pinv
  \]
  and this concludes.
\end{proof}
\begin{proof}[Proof of \Cref{lem:join-scott}]
  The function can also be obtained as the join in the dcpo $[\Xi \to
  \CC(X,Y)]$, it is therefore Scott continuous.
\end{proof}

\begin{proposition}
  \label{prop:rev-substitution-interpretation}
  Given a well-formed term $\Psi ; \Delta \entail t \colon A$ and for all $(x_i
  \colon A_i) \in \Delta$, a well-formed term $\Psi ; \emptyset \entail t_i \colon A_i$; if
  $\sigma = \set{x_i \mapsto t_i}_i$, then for all $g \in \sem\Psi$:
  \[
    \sem{\Psi ; \emptyset \entail \sigma(t) \colon A}(g) = \sem{\Psi ; \Delta
    \entail t \colon A}(g) \circ \left(\bigotimes_i \sem{\Psi ; \emptyset
    \entail t_i \colon A_i}(g) \right).
  \]
  We define then, for all $g \in \sem\Psi$: $\sem\sigma(g) = \bigotimes_i
  \sem{\entail t_i \colon A_i}(g)$.
\end{proposition}
\begin{proof}
  The proof is straightforward by induction on the formation rules for $t$.
\end{proof}

\begin{lemma}
  \label{lem:rev-matching-semantics}
  Given two well-formed values $\Psi ; \Delta \entail v \colon A$ and $\Psi ;
  \emptyset \entail v' \colon A$, and a substitution $\sigma$, if
  $\match{\sigma}{v}{v'}$ then for all $g \in \sem\Psi$:
  \[ \sem v(g)\pinv \circ
  \sem{v'}(g) = \sem\sigma(g). \]
\end{lemma}
\begin{proof}
  The proof is straightforward by induction on $\sigma[b]=b'$.
\end{proof}

\begin{lemma}
  \label{lem:rev-beta-sound}
  Given a well-formed iso abstraction $~\entailiso \isobreduit \colon A
  \iso B$ and a well-formed value $~\entail v' \colon A$, if
  $\match{\sigma}{v_i}{v'}$, then
  \[
    \sem{\entail \isobreduit~v' \colon B}
    = \sem{\entail \sigma(v_i) \colon B}.
  \]
\end{lemma}
\begin{proof}
  First, we deduce from the assumption $\match{\sigma}{v_i}{v'}$ that
  \begin{itemize}
    \item $\neg(v_i~\bot~v')$, and thus $\sem{v_i}\pinv \circ \sem{v'} =
      \sem\sigma$, thanks to Lemma~\ref{lem:rev-matching-semantics}.
    \item for all $j \neq i$, $v_j~\bot~v'$, and thus $\sem{v_j}\pinv \circ
      \sem{v'} = 0$, thanks to Lemma~\ref{lem:rev-orthogonal-semantics}.
  \end{itemize}
  We can then compute the semantics, with $\omega \defeq \isobreduit$:
  \begin{align*}
    &\ \sem{\omega~v'} & \\
    &= \sem\omega \circ \sem{v'} & \text{(by definition.)} \\
    &= \left( \bigvee_j \sem{v_j} \circ \sem{b_j}\pinv \right) \circ \sem{v'}
    & \text{(by definition.)} \\
    &= \bigvee_j \sem{v_j} \circ \sem{b_j}\pinv \circ \sem{v'}
    & \text{(composition distributes over join.)} \\
    &= \sem{v_i} \circ \sem{b_i}\pinv \circ \sem{v'}
    & \text{(lemma~\ref{lem:rev-orthogonal-semantics}.)} \\
    &= \sem{v_i} \circ \sem\sigma
    & \text{(lemma~\ref{lem:rev-matching-semantics}.)} \\
    &= \sem{\sigma(v_i)}
    & \text{(prop.~\ref{prop:rev-substitution-interpretation}.)}
  \end{align*}
\end{proof}

\smallskip
\noindent
\textsf{\bfseries Proof of Soundness.~}

\begin{lemma}
  \label{lem:rev-iso-sound}
  Given a well-formed iso judgement $~\entailiso \omega \colon T$,
  if $\omega \to \omega'$, then
  \[
    \sem{\entailiso \omega \colon T}
    = \sem{\entailiso \omega' \colon T}.
  \]
\end{lemma}
\begin{proof}
  The proof is done by induction on $\to$.
  \begin{itemize}
    \item $\ffix \phi . \omega \to \omega[\ffix \phi . \omega/\phi]$.
      \begin{align*}
        &\ \sem{\ffix \phi . \omega} & \\
        &= \fix \sem\omega
        & \text{(by definition.)} \\
        &= \sem\omega \circ \fix(\sem \omega)
        & \text{(fixed point.)} \\
        &= \sem \omega \circ \sem{\ffix \phi . \omega}
        & \text{(by definition.)} \\
        &= \sem{\omega[\ffix \phi . \omega/\phi]}
        & \text{(Lem~\ref{lem:sem-subst-iso}.)}
      \end{align*}
    \item $(\lambda \phi . \omega_1) \omega_2 \to \omega_1[\omega_2/\phi]$.
      \begin{align*}
        &\ \sem{(\lambda \phi . \omega_1) \omega_2} & \\
        &= \rmeval \circ \pv{\rmcurry(\sem{ \omega_1})}{\sem{\omega_2}}
        & \text{(by definition.)} \\
        &= \sem{\omega_1} \circ \sem{\omega_2}
        & \\
        &= \sem{\omega_1[\omega_2/\phi]}
        & \text{(Lem~\ref{lem:sem-subst-iso}.)}
      \end{align*}
    \item $\omega_1 \omega_2 \to \omega'_1 \omega_2$.
      Direct with the induction hypothesis.
  \end{itemize}
\end{proof}

\begin{proof}[Proof of \Cref{prop:rev-soundness}]
	The proof is done by induction on $\to$.
	\begin{itemize}
		\item $\isobreduit v' \to \sigma(v_j)$. This is covered by Lemma~\ref{lem:rev-beta-sound}.
		\item $\letv{p}{v}{t} \to \sigma(t)$. This is a conclusion of
			Prop.~\ref{prop:rev-substitution-interpretation}.
		\item $\omega~t \to \omega'~t$. We conclude with the previous lemma
			(Lemma~\ref{lem:rev-iso-sound}).
		\item $\ini t \to \ini t'$ when $t \to t'$. The induction hypothesis gives $\sem t = \sem{t'}$
			and then
			\[
				\sem{\ini t} = \iota_i \sem t = \iota_i \sem{t'} = \sem{\ini t'}.
			\]
		\item $t \otimes t_2 \to t' \otimes t_2$ when $t \to t'$.
			\[
				\sem{t \otimes t_2} = \sem{t} \otimes \sem{t_2} \overset{\text{IH}}{=}
				\sem{t'} \otimes \sem{t_2} = \sem{t' \otimes t_2}.
			\]
		\item $\isobreduit t \to \isobreduit t'$ when $t \to t'$. In general,
			\[
				\sem{\omega~t} = \sem{\omega} \sem t \overset{\text{IH}}{=} \sem\omega \sem{t'}
				= \sem{\omega~t'}.
			\]
		\item $\fold t \to \fold t'$ when $t \to t'$.
			\[
				\sem{\fold t} = \alpha \sem t \overset{\text{IH}}{=} \alpha \sem{t'} = \sem{\ini t'}.
			\]
		\item $\letv{p}{t}{t_2} \to \letv{p}{t'}{t_2}$ when $t \to t'$.
			\[
				\sem{\letv{p}{t}{t_2}} = \sem{t_2} (\iid \otimes \sem t)
				\overset{\text{IH}}{=} \sem{t_2} (\iid \otimes \sem{t'})
				= \sem{\letv{p}{t'}{t_2}}.
			\]
	\end{itemize}
\end{proof}

\begin{corollary}
  \label{cor:rev-soundness}
  Given a well-formed term judgement $~\entail t \colon A$ such that
  $t \downarrow$, then $\sem{\entail t \colon A} \neq 0_{\sem A}$.
\end{corollary}
\begin{proof}
  Knowing that there is a value $v$ such that $t \to^* v$,
  Prop.~\ref{prop:rev-soundness} ensures that $\sem t = \sem v \neq 0$.
\end{proof}

\section{Proof of Adequacy}
\label{sub:rev-proof-adequacy}

Our proof of adequacy involves a finitary sublanguage, where the
number of recursive calls is controlled syntactically. We show the
adequacy result for the finitary terms thanks to strong normalisation,
and then show that it implies adequacy for the whole language; this is
achieved by observing that a normalising finitary term is also
normalising in its non-finitary form.

\noindent
\textsf{\bfseries Finitary sublanguage.~}
We introduce the syntax for finitary terms, where the number of possible
reductions is limited by the syntax itself. The grammar of finitary isos
is given by:
\[
	\omega~::=~\isobreduit \alt \lambda \phi . \omega \alt \phi \alt \omega
	\omega \alt \nfix n \phi . \omega
\]
where $n$ is a natural number. The iso $\nfix n \phi . \omega$ has the same
formation rule as $\ffix \phi . \omega$ above. We introduce syntactic sugar,
that denotes an expression that will never reduce, by induction on iso types:
\[
	\begin{array}{c}
		\Omega_{A \iso B} \defeq \set{\mid \cdot}
		\qquad
		\Omega_{T_1 \to T_2} \defeq \lambda \phi^{T_1} . \Omega_{T_2}
	\end{array}
\]
The syntax of finitary terms does not change compared to the language presented
at the beginning of the chapter, with the addition of a term $\bot$, which
indicates that there was no match when applying an iso abstraction to a value.
The finitary operational semantics is then defined as:
\[
	\begin{array}{c}
		\infer{\nfix 0 \phi^ T . \omega \finto \Omega_T}{}
		\qquad
		\infer{\nfix{n+1} \phi . \omega \finto \omega[\nfix n \phi . \omega /
		\phi]}{}
		\mynl
		\infer{(\lambda \phi. \omega_1) \omega_2 \finto \omega_1[\omega_2/\phi]}{}
		\qquad
		\infer{\omega_1\omega_2 \finto \omega'_1 \omega_2}{\omega_1 \finto \omega'_1}
	\end{array}
\]
\[
	\begin{array}{c}
		\infer{ \isobreduit v' \finto \sigma(e_i)}{
			\match{\sigma}{v_i}{v'}}
		\qquad
		\infer{ \isobreduit v' \finto \bot}{
			\forall i, \neg(\match{\sigma}{v_i}{v'})}
		\\[1.5ex]
		\infer{C_\to[t_1] \finto C_\to[t_2]}{t_1 \finto t_2}
		\qquad
		\infer{C_\to[t] \finto \bot}{t \finto \bot}
		\qquad
		\infer{\letv{p}{v}{t} \finto \sigma(t)}{\match{\sigma}{p}{v}}
		\qquad
		\infer{\omega~t \finto \omega'~t}{\omega \finto \omega'}
	\end{array}
\]

\begin{lemma}[Iso Subject Reduction]
	\label{lem:fin-iso-subject-reduction}
	If $\Psi \entailiso \omega \colon T$ is well-formed, $\omega$ is finitary and
	$\omega \finto \omega'$, then $\Psi \entailiso \omega' \colon T$.
\end{lemma}

\begin{lemma}[Iso Progress]
	\label{lem:iso-fin-progress}
	If $\Psi \entailiso \omega \colon T$ is well-formed and $\omega$ is finitary,
	$\omega$ is either an iso value or there exists $\omega'$ such that $\omega
	\finto \omega'$.
\end{lemma}

\begin{lemma}[Subject Reduction]
	\label{lem:fin-subject-reduction}
	If $\Psi; \Delta \entail t \colon A$ is well-formed, $t$ is finitary and $t
	\finto t'$, then $\Psi; \Delta \entail t' \colon A$ is also well-formed.
\end{lemma}

\begin{lemma}[Progress]
	\label{lem:fin-progress}
	If $~\entail t \colon A$ and $t$ is finitary, then:
	\begin{itemize}
		\item either $t$ is a value,
		\item or $t \to \bot$,
		\item or there exists $t'$ such that $t \finto t'$.
	\end{itemize}
\end{lemma}
\begin{proof}
	The proof is done by induction on $~\entail t \colon A$.
	\begin{itemize}
		\item $~\entail \unit \colon \one$ is a value.
		\item $~\entail t_1 \otimes t_2 \colon A \otimes B$.
			By induction hypothesis, either $t_1$ reduces, in which case $t_1 \otimes
			t_2$ reduces too, or $t_1$ is a value. If $t_1$ is a value, by induction
			hypothesis, either $t_2$ reduces, in which case $t_1 \otimes t_2$
			reduces, or $t_2$ is a value, and thus $t_1 \otimes t_2$ is a value.
		\item $\Psi ; \Delta \entail \ini t \colon A_1 \oplus A_2$.
			By induction hypothesis, either $t$ reduces, in which case $\ini t$ too,
			or $t$ is a value, and thus $\ini t$ is a value.
		\item $\Psi ; \Delta \entail \fold t \colon \mu X . A$.
			By induction hypothesis, either $t$ reduces, in which case $\fold t$ too,
			or $t$ is a value, and thus $\fold t$ is a value.
		\item $\Psi ; \Delta \entail \omega~t \colon B$.
			Thanks to Lemma~\ref{lem:iso-fin-progress}, $\omega$ either reduces, in which
			case $\omega~t$ also reduces, or it is an iso value $\isobreduit$. In the last case,
			in the induction hypothesis gives that either $t$ reduces, in which case
			$\omega~t$ also reduces, or $t$ is value. In that case, either $t$ matches with
			one $v_i$, and $\omega~t$ reduces, or it does not, and $\omega~t \to \bot$.
		\item $\Psi ; \Delta_1, \Delta_2 \entail \letv{p}{t_1}{t_2} \colon B$.
			In both cases of the induction hypothesis, $\letv{p}{t_1}{t_2}$ reduces.
	\end{itemize}
\end{proof}

\noindent
\textsf{\bfseries Strong Normalisation.~}
We prove that the reduction $\finto$ is strongly normalising, by observing that
this system is separated in two very distinct systems: one that
reduces the iso $\lambda$-terms, and another that performs the reversible
computations. We show that both those systems can be extended to commute with
each other, which ensures strong normalisation as long as both are strongly
normalising. We start by introducing the system that performs the reductions on
isos.

\[
	\begin{array}{c}
		\infer{\nfix 0 \phi^ T . \omega \isoto \Omega_T}{}
		\qquad
		\infer{\nfix{n+1} \phi . \omega \isoto \omega[\nfix n \phi . \omega /
		\phi]}{}
		\mynl
		\infer{(\lambda \phi. \omega_1) \omega_2 \isoto \omega_1[\omega_2/\phi]}{}
		\qquad
		\infer{\omega_1\omega_2 \isoto \omega'_1 \omega_2}{\omega_1 \isoto \omega'_1}
	\end{array}
\]
\[
	\begin{array}{c}
		\infer{
			\isobreduit \isoto \left\{
				\begin{array}{ll}
					\alt v_i \iso e_i & \text{ if } i \neq j \\
					\alt v_j \iso e'_j & \text{ else}
				\end{array}
			\right\}
		}{
			e_j \isoto e'_j
		}
		\mynl
		\mynl
		\infer{C_\to[t_1] \isoto C_\to[t_2]}{t_1 \isoto t_2}
		\qquad
		\infer{\letv{p}{t'}{t_1} \isoto \letv{p}{t'}{t_2}}{t_1 \isoto t_2}
		\qquad
		\infer{\omega~t \isoto \omega'~t}{\omega \isoto \omega'}
	\end{array}
\]

\begin{lemma}
	\label{lem:iso-sn}
	The reduction system $\isoto$ is strongly normalising.
\end{lemma}
\begin{proof}
	We translate finitary isos and finitary terms into a simply-typed
	$\lambda$-calculus with pairs. We write $\abs t$ the translation of $t$.
	\[
		\begin{array}{c}
			\abs \unit \defeq \unit
			\qquad
			\abs x \defeq \unit
			\qquad
			\abs{\ini t} \defeq \abs t
			\qquad
			\abs{\pv{t}{t'}} \defeq \pv{\abs t}{\abs{t'}}
			\mynl
			\abs{\fold t} \defeq \abs t
			\qquad
			\abs{\omega~t} \defeq \pv{\abs\omega}{\abs t}
			\qquad
			\abs{\letv{p}{t}{t'}} \defeq \pv{\abs t}{\abs{t'}}
			\mynl
			\abs{\isobreduit} \defeq \langle \abs{e_i} \rangle _{i \in I}
			\qquad
			\abs{\lambda \phi . \omega} \defeq \lambda \phi . \abs \omega
			\qquad
			\abs \phi \defeq \phi
			\mynl
			\abs{\omega_2\omega_1} \defeq \abs{\omega_2} \abs{\omega_1}
			\qquad
			\abs{\nfix n \phi . \omega} \defeq \nfix n \phi . \abs \omega
		\end{array}
	\]
	This $\lambda$-calculus is strong normalising; this can be proven with
	candidates of reducibility, \emph{à la} System F \cite[Chapters 11 and
	14]{girard1989proofs}.
\end{proof}

We then introduce the reductions that perform the reversible computation,
our equivalent of $\beta$-reduction but for our iso language.
\[
	\begin{array}{c}
		\infer{ \isobreduit v' \termto \sigma(e_i)}{
			\match{\sigma}{v_i}{v'}}
		\qquad
		\infer{ \isobreduit v' \termto \bot}{
			\forall i, \neg(\match{\sigma}{v_i}{v'})}
		\\[1.5ex]
		\infer{C_\to[t_1] \termto C_\to[t_2]}{t_1 \termto t_2}
		\qquad
		\infer{C_\to[t] \termto \bot}{t \termto \bot}
		\qquad
		\infer{\letv{p}{v}{t} \termto \sigma(t)}{\match{\sigma}{p}{v}}
	\end{array}
\]
This system is strongly normalising thanks to a decreasing argument:
the number of isos and $\mathtt{let}$ constructors strictly decreases
when applying the reduction $\termto$.
\begin{lemma}
	\label{lem:term-red-sn}
	The reduction system $\termto$ is strongly normalising.
\end{lemma}

\begin{lemma}
	\label{lem:red-commutative}
	$\termto~\isoto~\subseteq~\isoto~\termto$.
\end{lemma}

We say that $\isoto$ commutes \cite{bachmair1986commutation} over $\termto$.
This and the strong normalisation of both systems $\isoto$ and $\termto$
ensures the strong normalisation of them combined $\isoto \cup \termto$
\cite[Theorem 1]{bachmair1986commutation}.

\begin{lemma}
	\label{lem:sub-system-fin}
	$\finto~\subseteq~\isoto \cup \termto$.
\end{lemma}
\begin{proof}
	The proof is direct, by showing that any rule in $\finto$ is either
	in $\isoto$ or $\termto$.
\end{proof}

\begin{theorem}
	\label{th:fin-sn}
	The reduction system $\finto$ is strongly normalising.
\end{theorem}
\begin{proof}
	With \cite[Theorem 1]{bachmair1986commutation} and Lemmas~\ref{lem:iso-sn},
	\ref{lem:term-red-sn}, and~\ref{lem:red-commutative}, we can ensure that
	$\isoto \cup \termto$ is strongly normalising. We conclude then with
	Lemma~\ref{lem:sub-system-fin}, that shows that $\finto$ is a subsystem of a
	strongly normalising system.
\end{proof}

\noindent
\textsf{\bfseries Finitary adequacy.~}
We prove adequacy, but for finitary terms. To do so, we also need to introduce
the denotational semantics of finitary isos. The interpretation of $\nfix n$,
instead of being Kleene's fixed point, is the morphism obtained by unfolding
$n$ times. The interpretation of $\Omega_T$ is the bottom element of $\sem T$.

\begin{lemma}
	\label{lem:fin-iso-sound}
	Given a well-formed finitary iso judgement $~\entailiso \omega \colon T$, if
	$\omega \finto \omega'$, then
	\[
		\sem{\entailiso \omega \colon T}
		= \sem{\entailiso \omega' \colon T}.
	\]
\end{lemma}

\begin{proposition}[Finitary Soundness]
	\label{prop:fin-soundness}
	Given a well-formed finitary term judgement $~\entail t \colon A$, if $t
	\finto t'$, then
	\[
		\sem{\entail t \colon A} = \sem{\entail t' \colon A}.
	\]
\end{proposition}

\begin{theorem}[Finitary Adequacy]
	\label{th:fin-adeq}
	Given a well-formed finitary term judgement $~\entail t \colon A$,
	$t \downarrow$ iff $\sem{\entail t \colon A} \neq 0_{\sem A}$.
\end{theorem}
\begin{proof}
	We prove both directions of the double implication.
	\begin{itemize}
		\item[$(\Rightarrow)$] Knowing that $t \downarrow$, there exists a value
			$v$ such that $t \finto^* v$, and Prop.~\ref{prop:fin-soundness} ensures
			that $\sem t = \sem v \neq 0$.
		\item[$(\Leftarrow)$] We know that $\finto$ is strongly normalising (see
			Th.~\ref{th:fin-sn}), which means that the reduction from $t$ terminates,
			and Lemma~\ref{lem:fin-progress} ensures that it terminates either on a
			value $v$ or on $\bot$. However, $\sem t \neq 0$, thus it cannot
			terminate on $\bot$ because of Prop.~\ref{prop:fin-soundness}. We have
			then $t \finto^* v$, which concludes.
	\end{itemize}
\end{proof}

\noindent
\textsf{\bfseries Finitary Subterms.~}
We conclude in two steps. First we observe that the interpretation of a term is
nothing more than the join of the interpretations of its finitary subterms,
then we show that a reduction $\to^*$ can be linked to a finitary reduction
$\finto^*$.

\begin{definition}[Finitary Subiso]
	\label{def:fin-subiso}
	Let $\subfin$ be the smallest relation between (finitary or not) isos such that:
	\[
		\begin{array}{c}
			\infer{\nfix n \phi . \omega \subfin \ffix \phi . \omega}{}
			\qquad
			\infer{\omega[\omega_1/\phi] \subfin
			\omega[\omega_2/\phi]}{\omega_1 \subfin \omega_2}
		\end{array}
	\]
\end{definition}

\begin{lemma}
	\label{lem:subfin-to-dcpo}
	Given two well-formed (finitary or not) iso judgements
	$\Psi \entailiso \omega_1 \colon T$ and $\Psi \entailiso \omega_2 \colon T$
	such that $\omega_1 \subfin~\omega_2$, then
	\[
		\sem{\Psi \entailiso \omega_1 \colon T}
		\leq \sem{\Psi \entailiso \omega_2 \colon T}.
	\]
\end{lemma}

\begin{lemma}
	\label{lem:subfin-join}
	Given a well-formed iso judgement $\Psi \entailiso \omega \colon T$,
	we have:
	\[
		\sem{\Psi \entailiso \omega \colon T} =
		\bigvee_{\underset{\omega' \text{ finitary}}{\omega'\subfin~\omega}}
		\sem{\Psi \entailiso \omega' \colon T}.
	\]
\end{lemma}
\begin{proof}
	We observe that, by definition:
	\[
		\sem{\Psi \entailiso \ffix \phi . \omega \colon T}
		= \bigvee_{n \in \N} \sem{\Psi \entailiso \nfix n \phi . \omega \colon T}
	\]
	which proves the desired result in the case of $\ffix \phi . \omega$. The
	general conclusion falls by induction.
\end{proof}

We generalise to terms the definition of subisos given above.

\begin{definition}[Finitary Subterm]
	\label{def:fin-subterm}
	Let $\subfin$ be the smallest congruence relation between (finitary or not)
	terms such that:
	\[
		\begin{array}{c}
			\infer{\omega_1~t \subfin~\omega_2~t}{\omega_1 \subfin~\omega_2}
		\end{array}
	\]
\end{definition}

The following lemma follows directly from the previous definition and
Lemma~\ref{lem:subfin-join}.

\begin{lemma}
	\label{lem:subfin-join-term}
	Given a well-formed term judgement $\Psi; \Delta \entail t \colon A$,
	we have:
	\[
		\sem{\Psi; \Delta \entail t \colon A} =
		\bigvee_{\underset{t' \text{ finitary}}{t'\subfin~t}}
		\sem{\Psi; \Delta \entail t' \colon A}.
	\]
\end{lemma}

It is also the right time to observe that if a term has a finitary subterm that
reduces to a value eventually, the former also normalises to the same value.

\begin{lemma}
	\label{lem:fin-red-to}
	Given a well-formed closed term judgement $~\entail t \colon A$, if
	there exists a finitary subterm $t' \subfin t$ and a value such that
	$t' \finto^* v$, then $t \to^* v$.
\end{lemma}
\begin{proof}
	The finitary term $t'$ has the same reduction steps as $t$ up to a point.
	Lemma~\ref{lem:fin-progress} ensures that this end point is either a value or
	$\bot$ in the finitary case. Thus if the reduction from $t'$ gets to a value,
	the reduction from $t$ must also finish on a value. Since the reduction steps
	were exactly the same, both reductions have the same normal form.
\end{proof}

\noindent
\textsf{\bfseries Conclusion.~}
We finally have all the tools to conclude with adequacy for closed terms of our
original language.

\begin{proof}[Proof of Adequacy (Theorem~\ref{th:rev-adeq})]
	There are two implications to prove.
	\begin{itemize}
		\item[$(\Rightarrow)$] This first implication is proven as
			Corollary~\ref{cor:rev-soundness}.
		\item[$(\Leftarrow)$] Suppose that $\sem t \neq 0$. Necessarily, thanks to
			Lemma~\ref{lem:subfin-join-term}, there exists a finitary term $t'$ such
			that $t' \subfin t$ and $\sem{t'} \neq 0$. In Theorem~\ref{th:fin-adeq},
			we have proven adequacy for finitary terms, meaning that there exists a
			value $v$ such that $t' \finto^* v$. Lemma~\ref{lem:fin-red-to} ensures
			then that $t \to^* v$, which concludes.
	\end{itemize}
\end{proof}

\section{Proofs of Section~\ref{sec:sem-preservation}}

\begin{proof}[Proof of \Cref{lem:duplication-typed}]
  The first case if done by a straightforward case analysis on $A$.
  The second case is done by induction on $A$. the case of the type
  $\one, A\otimes B$ and $A\oplus B$ are straightforward. For an
  inductive type $\mu X. A$, due to~\Cref{lem:duplication-invariant}
  we can reason on whether or not $X$ is in the set $S$. If it is, it
  means that a previous instance of the type $\mu X.  A$ has been
  encountered, in which case the associated iso variable $\isovar$ is
  of type $\mu X. A \iso (\mu X. A) \otimes (\mu X. A)$ and so the iso
  is well-typed. If $X$ is not in $S$ then by induction hypothesis.
\end{proof}

\begin{proof}[Proof of \Cref{lem:duplication-semantics}]
  The reduction is proven by a straightforward induction on $v$.
	The denotational semantics is given by definition.
\end{proof}

\begin{proof}[Proof of \Cref{lem:garbage-removal-typed}]
  Direct as each variable is used indeed only once, and by well-typing
  hypothesis on $\omega, \omega'$, \Cref{lem:inv-type}
  and~\Cref{lem:duplication-typed}.
\end{proof}

\begin{proof}[Proof of \Cref{lem:garb-remove-sem}]
  Direct by~\Cref{lem:duplication-semantics}.
\end{proof}

\begin{proof}[Proof of Theorem~\ref{th:computable}]
  \begin{itemize}
    \item In $\omega$, call the right-hand-side $e$.
    \item Notice that $\interp{e} = \interp{\floor{-}_B^{-1}} \circ
			\interp{\garRem{\opn{isos}(M_g)}{\opn{isos}(M_g')}} \circ \interp{\floor{-}_A}$
    \item By~\Cref{thm:rtm-iso-simulation} and \Cref{prop:rev-soundness}, we know that
    $\interp{\garRem{\opn{isos}(M_g)}{\opn{isos}(M_g')}} = g$.
    \item Therefore, since $g = \interp{\floor{-}_B} \circ f \circ
			\interp{\floor{-}_A^{-1}}$ by definition, we get

			\begin{align*}
				\sem e &= \interp{\floor{-}_B^{-1}} \circ
				\interp{\garRem{\opn{isos}(M_g)}{\opn{isos}(M_g')}} \circ \interp{\floor{-}_A} \\
				&= \interp{\floor{-}_B^{-1}} \circ g \circ \interp{\floor{-}_A} \\
				&= \interp{\floor{-}_B^{-1}} \circ \interp{\floor{-}_B} \circ f \circ
    \interp{\floor{-}_A^{-1}} \circ \interp{\floor{-}_A}
			\end{align*}

    \item By~\Cref{prop:rev-soundness} we get that $\interp{\floor{-}_B^{-1}} \circ
    \interp{\floor{-}_B} = \iid_B$ and $\interp{\floor{-}_A^{-1}} \circ
    \interp{\floor{-}_A} = \iid_A$.

    \item Therefore $\interp{e} = f$.

    \item Since the left-hand-side of $\omega$ is just a variable we get
    $\interp{\omega} = \interp{e}\circ \iid\inv = \interp{e} = f$.
  \end{itemize}
\end{proof}

\end{document}